%% file: main.tex
\documentclass[12pt,a4paper]{article}
\usepackage[table,dvipsnames]{xcolor}
\usepackage{amsmath,amssymb,amsthm}
\usepackage{graphicx,psfrag,epsf,caption}
\usepackage[round]{natbib}

\usepackage[labelformat=simple]{subcaption}

\usepackage{enumerate,array}
\usepackage{multirow,setspace,titlesec,tikz}
\usetikzlibrary{arrows}
\usepackage{bclogo,xr}
\usepackage{booktabs,float} 
\usepackage{url} 
\usepackage{hyperref} 
\usepackage[normalem]{ulem}
\usepackage{pifont}
\usepackage{bm}
\usepackage{wrapfig}

\usepackage{authblk} 
\usepackage{pdflscape} 

\newcommand{\bX}{\boldsymbol{X}}
\newcommand{\bx}{\boldsymbol{x}}

\newcommand{\bR}{\boldsymbol{R}}
\newcommand{\bS}{\boldsymbol{S}}
\newcommand{\bW}{\boldsymbol{W}}
\newcommand{\bI}{\boldsymbol{I}}
\newcommand{\bT}{\boldsymbol{T}}
\newcommand{\balpha}{\boldsymbol{\alpha}}

\newcommand{\btheta}{{\boldsymbol{\theta}}}
\newcommand{\bTheta}{\boldsymbol{\Theta}}

\newcommand{\bmu}{\boldsymbol{\mu}}

\newcommand{\bnu}{\boldsymbol{\nu}}

\newcommand{\cB}{\mathcal{B}}

\newcommand{\cS}{\mathcal{S}}
\newcommand{\cG}{\mathcal{G}}
\newcommand{\cN}{\mathcal{N}}

\newcommand{\bbE}{\mathbb{E}}

\newcommand{\bbV}{\mathbb{V}}

\newcommand{\xmark}{\ding{55}}%

\newtheorem{lemma}{Lemma}

\newtheorem{theorem}{Theorem}

\graphicspath{{figs/}}

\newcommand{\blind}{1}

\begin{document}

\def\spacingset#1{\renewcommand{\baselinestretch}%
{#1}\small\normalsize} \spacingset{1}


\if1\blind
{
  \title{\bf Structure Learning for Hybrid Bayesian Networks}

\author[1,2]{Wanchuang Zhu}
\author[2,3]{Ngoc Lan Chi Nguyen}
\affil[1]{Data61, CSIRO, Australia}
\affil[2]{ARC Centre for Data Analytics for Resources and Environments, Australia}
\affil[3]{School of Computer Science, The University of Sydney, Australia}

\date{}
\setcounter{Maxaffil}{0}
\renewcommand\Affilfont{\itshape\small}
  \maketitle
} \fi

\if0\blind
{
  \bigskip
  \bigskip
  \bigskip
  \begin{center}
    {\LARGE\bf Structure Learning for Hybrid Bayesian Networks}
\end{center}
  \medskip
} \fi

\bigskip

\begin{abstract}
Bayesian networks have been used as a mechanism to represent the joint distribution of multiple random variables in a flexible yet interpretable manner. One major challenge in learning the structure of a Bayesian network is how to model networks which include a mixture of continuous and discrete random variables, known as hybrid Bayesian networks. This paper overviews the literature on approaches to handle hybrid Bayesian networks.
Typically one of two approaches is taken: either the data are considered to have a joint distribution which is designed for a mixture of discrete and continuous variables, or continuous random variables are discretized, resulting in discrete Bayesian networks. In this paper, we propose a strategy to model all random variables as Gaussian, referred to it as  {\it Run it As Gaussian (RAG)}. We demonstrate that RAG results in more reliable estimates of graph structures theoretically and by simulation studies, than converting continuous random variables to discrete. 
Both strategies are also implemented on a childhood obesity data set. The two different strategies give rise to significant differences in the optimal graph structures, with the results of the simulation study suggesting that our strategy is more reliable.
\end{abstract}

\section{Introduction} \label{sec:intro}
Bayesian networks (BNs) have emerged as a useful framework to conceptualize and model the relationships among many interacting random variables. The applicability of BNs spans many fields, including but not limited to, gene regulatory networks~\citep{xing2017GRN}, demographics~\citep{sun2015bayesian}, risk management~\citep{tavana2018artificial}, and psychology~\citep{ van2017systematic}.

The structure of a BN represents graphically the joint distribution of the random variables under consideration \citep{koller2009probabilistic}. Random variables and the connections between them are represented as nodes and edges respectively. The graph structures of these networks are referred to as directed acyclic graphs (DAG). In a DAG, a directed edge between random variables indicates probabilistic dependency. The node from which an edge originates is called a parent, while the node to which an edge points is called a child. It is important to note that for a DAG, a variable is conditionally independent of all its predecessors given the state of its parents. Therefore, the joint probability distribution of the random variables can be expressed as a product of conditionally independent distributions of all variables, each of which depends only on its parents. Structure learning refers to learning the structure of the network, which is determined by the presence and direction of edges.

A well known challenge in structure learning is that the number of possible structures grows super-exponentially with respect to the number of random variables in the network \citep{1978DAGn,kuipers2017partition}. For example, a network with 15 variables, can form up to $10^{43}$ possible DAGs. The posterior distribution of the DAGs is discrete, making approximations relying on gradients, such as variational Bayes, inaccurate \citep{zhang2018advances}. Therefore, computationally intensive sampling methods, such as Markov chain Monte Carlo (MCMC), are often the only means of estimating the posterior surface. Overviews of the challenges and approaches for structure learning of Bayesian networks have been thoroughly discussed by \cite{heinze2018causal,glymour2019review,vowels2021d}.

A neglected difficulty in the literature is structure learning for \textit{Hybrid Bayesian networks} (HBNs) which contains both continuous and discrete variables. The existing approaches avoid such challenge by largely assuming that the variables are either all continuous or all discrete. For example, the score-based methods for structure learning define a score function that measures the structure's fitness to the observed data and then employ a search algorithm over the set of potential network structures. Typical score functions include K2 score, 
likelihood score, normalized minimum likelihood, AIC or BIC score and Bayesian score \citep{koller2009probabilistic}. All these score functions can not be directly applied to HBNs. Although new score functions were actively proposed, including \cite{zheng2018dags,zhang2019d,yu2019dag} and \cite{lorch2021dibs}, the challenge in HBNs has not properly tackled.

The limited researches for structure learning in HBNs can be summarised into two categories: it either develops a score function that can directly model the dependency between discrete and continuous variables, or it discretizes continuous random variables, thus making all variables in the network discrete. Examples in the first category include \cite{lauritzen1989graphical,geiger1994learning,bach2002learning,moral2001mixtures} and \cite{sokolova2014causal}.
However, the modeling between continuous and discrete variable is computationally heavy and tedious. To reduce the computational burden, some approaches have to impose restrictive assumptions on their models. In particular, they do not allow continuous parent nodes to have discrete nodes as their child nodes. For example, the classic conditional linear Gaussian model (CLG) suffers from the above constraint \citep{lauritzen1989graphical, lauritzen1992propagation}.

The strategy which transforms continuous random variables into discrete ones via discretization suffers from information loss. Commonly used methods involve dividing continuous variables into equal intervals, or into equal quantiles, or the interval lengths are varied according to expertise from a specific field \citep{nojavan2017comparative}. Researchers have made efforts to reduce the loss of information during discretization. \cite{rissanen1978modeling} and \cite{friedman1996discretizing} attempt to minimise information loss is via the use of minimum description length (MDL) principle. \cite{friedman1996discretizing} aimed to calculate the ideal number of intervals for discretization and edge locations when dealing with continuous variables. The selected discretization policy is the one that provides MDL for the discretized Bayesian network, while requiring the least amount of information to recover the original continuous values. 
\cite{monti2013multivariate} used a Bayesian score metric to measure goodness of discretization policies. However, the Bayesian score metric relies on an unwarranted assumption that each continuous variable is a noisy observation of an underlying discrete variable. In addition, the discretization policy needs to be updated dynamically corresponding to the changes in DAG structures, which is  computationally heavy and makes interpreting the dicretization policy difficult.
\cite{chen2017learning} extend single-variable discretization techniques \citep{boulle2006modl, lustgarten2011application} by modifying a prior that can reduce the cubic complexity to quadratic, although the applicability of the prior choices for discretization policies requires further investigation.

Given that the existing approaches suffer from computational issues, in this paper, we propose a strategy to learn structures in hybrid Bayesian networks. We show that considering all discrete random variables as continuous outperforms techniques which seek to accommodate the discreteness, both from a theoretical standpoint as well as empirical evidence. Such a strategy not only avoids the need to model the complex relationship between continuous and discrete variables, but also ensures no information loss. The rest of the paper is structured as follows. In Section \ref{sec:preliminaries}, notations and MCMC schemes for structure learning are introduced. Theoretical properties are provided in Section \ref{sec:theory} for the proposed strategy under a score-based approach. Simulation studies are conducted in Section \ref{sec:simulations}. A real example is demonstrated in Section \ref{sec:realdata}, followed by discussion and conclusion in Section \ref{sec:conclusion}.

\section{Preliminaries on Bayesian Networks} \label{sec:preliminaries}
\subsection{Notations}

The two components of a Bayesian network, for a set of $n$ random variables $\bX = (X_1, \dots, X_n)$, are (1) a network structure with an underlying directed acyclic graph $\mathcal{G}$, in which each node represents a random variable $X_i$ and the edges represent directed dependencies between random variables; and (2) the conditional probability distributions $P(X_i\mid\textbf{Pa}_i^{\cG}, \btheta_i^{\cG})$ for each variable $X_i$ given a set of parameters $\btheta_i^{\cG}$ and the set of parents $\textbf{Pa}_i^{\cG}$, where $\textbf{Pa}_i^{\cG}$ denotes the set of parents for node $X_i$ in $\cG$. Given the graph structure $\mathcal{G}$ and corresponding parameters $\bTheta^{\cG}=(\btheta_1^{\cG},\cdots,\btheta_n^{\cG})$, the joint distribution of $\bX$ can be written as a product of conditionally independent distributions 
\setlength{\belowdisplayskip}{1pt} \setlength{\belowdisplayshortskip}{1pt}
\setlength{\abovedisplayskip}{1pt} \setlength{\abovedisplayshortskip}{1pt}
\begin{align*}
    P(X_1, \dots, X_n \mid \cG, \bTheta^{\cG}) = \prod_{i=1}^{n} P(X_i\mid\textbf{Pa}_i^{\cG}, \btheta_i^{\cG}),
\end{align*}
where the random variable in $\bX$ can be either discrete or continuous. Structure learning aims to discover the underlying true graph structure $\mathcal{G}$. Note that, different $\mathcal{G}$'s can explain the data $\bm X$ equally well and the structure of the network is only identifiable up to an equivalence class. An equivalence class consists of graphs that share the same skeleton and v-structure \citep{koller2009probabilistic}.
\subsection{MCMC for Structure Learning}
For a score-based approach in a Bayesian framework, the posterior probability of a DAG $\mathcal{G}$ given the data $\bX$ is used as the score function for the graphical structure. In particular,
\begin{align*}
    P(\mathcal{G}\mid\bX) \propto P(\bX\mid\mathcal{G})P(\mathcal{G})=\int P(\bX\mid \bTheta^{\cG},\mathcal{G})P(\bTheta^{\cG} \mid \mathcal{G}) d\bTheta^{\cG},
\end{align*}
where $P(\mathcal{G})$ denotes a prior distribution over the structures and $P(\bX\mid\mathcal{G})$ refers to the marginal likelihood, derived by marginalizing over the parameter $\bTheta^{\cG}$ and $P(\bTheta^{\cG} \mid \mathcal{G})$ is a prior distribution for parameters $\bTheta^{\cG}$ attached to the graph $\mathcal{G}$. It has been shown in \cite{geiger2002parameter} that the integration can be only validly done in two scenarios: the random variables either follow a multivariate normal distributions, or follow a multivariate multinomial distribution.

MCMC-based structure learning methods originated from \cite{madigan1995bayesian} and is known as {\it structure MCMC}. The acceptance rate in Metropolis-Hastings step is given as,
\begin{eqnarray*}
\alpha = \min \left\{1, \frac{P(\bX\mid\mathcal{G}^p)P(\mathcal{G}^p) q(\mathcal{G}^c \mid \mathcal{G}^p)}{P(\bX\mid\mathcal{G}^c)P(\mathcal{G}^c) q(\mathcal{G}^p \mid \mathcal{G}^c)} \right\},
\end{eqnarray*}
where $\mathcal{G}^c$ denotes the current graph, $\mathcal{G}^p$ denotes the proposed graph and $q(\mathcal{G}^c \mid \mathcal{G}^p)$ denotes the proposal distribution between two graphs.
Subsequent works sought to improve upon mixing and convergence, notably the introduction of {\it order MCMC} \citep{friedman2003structuremcmc} by building a Markov chain on node orderings, each of which covers a large set of DAGs, at the expense of bias in the resulting posterior distribution. To avoid this bias while maintaining reasonable convergence, \cite{grzegorczyk2008improving} augmented the classical structure MCMC with a new edge reversal move. \cite{kuipers2017partition} further proposed {\it partition MCMC}, which reduces the sampling space by collapsing the space of DAGs to a space of ordered partitions regarding the nodes. Partition MCMC is adopted as the default framework in this paper to draw posterior samples of DAG and evaluate performance of different strategies dealing with hybrid Bayesian networks.


We start with the description on how to derive an ordered partition from a DAG. The acyclic characteristic of DAGs means that they admit at least one outpoint: a node without any parent. 
By recursively taking out all the outpoints from a DAG, we obtain an ordered partition $\Lambda$. It contains two components: a permutation of nodes and a partition vector. Figure \ref{fig:example.partition} shows an illustrative example to derive the ordered partition from a DAG.

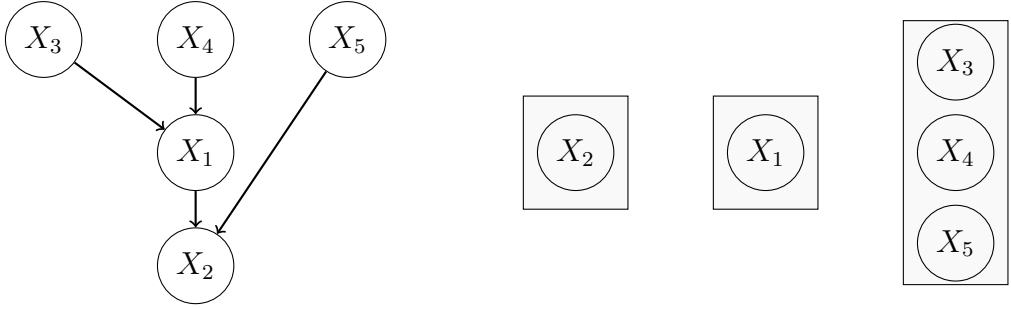
\begin{figure}[tb]
    \centering
    \begin{tikzpicture}
    \node[shape=circle,draw=black] (x3) at (-2,1.5) {$X_3$};
    \node[shape=circle,draw=black] (x4) at (0,1.5) {$X_4$};
    \node[shape=circle,draw=black] (x5) at (2,1.5) {$X_5$};
    \node[shape=circle,draw=black] (x1) at (0,0) {$X_1$};
    \node[shape=circle,draw=black] (x2) at (0,-1.5) {$X_2$};
    \draw[thick,->] (x3) to  (x1);
    \draw[thick,->] (x1)--(x2);
    \draw[thick,->] (x5)--(x2);
    \draw[thick,->] (x4)--(x1);
    
    \node[shape=rectangle,draw=black,fill = gray!5,minimum height = 1.5cm] (A.background) at (7.5,0) {\ \ \ \  \ \ \ \ \ };
    \node[shape=circle,draw=black] (A) at (7.5,0) {$X_1$};
    
    \node[shape=rectangle,draw=black,fill = gray!5,minimum height = 1.5cm] (B.background) at (5,0) {\ \ \ \  \ \ \ \ \ };
    \node[shape=circle,draw=black] (B) at (5,0) {$X_2$};
    \node[shape=rectangle,draw=black,fill = gray!5,minimum height = 3.5cm] (D.background) at (10,0) {\ \ \ \  \ \ \ \ \ };
    \node[shape=circle,draw=black] (C) at (10,1.2) {$X_3$};
    \node[shape=circle,draw=black] (D) at (10,0) {$X_4$};
    \node[shape=circle,draw=black] (D) at (10,-1.2) {$X_5$};
    \end{tikzpicture}
    \caption{A DAG graph and its corresponding ordered partition. The permutation $\prec = (X_2, X_1, X_3, X_4, X_5)$ and the partition vector $\Pi_\prec = [1,1,3]$.}
    \label{fig:example.partition}
\end{figure}


The mapping between the DAG space and partition space is not one-to-one mapping, since obviously one particular partition is compatible with multiple DAGs. 
Thus, the posterior distribution of ordered partition is given as,
\begin{eqnarray} \label{eqn:score.partition}
P(\Lambda\mid \bX) = \sum_{\mathcal{G}}P(\Lambda\mid\mathcal{G},\bX)P(\mathcal{G}\mid\bX) = \sum_{\mathcal{G}\in\Lambda}P(\mathcal{G}\mid \bX). 
\end{eqnarray}
Metropolis Hastings algorithm can be used to draw samples from posterior distribution shown in Equation \ref{eqn:score.partition}. The advantages of partition MCMC include that: (i) the space of partitions is much smaller than the DAG space, and (ii) a partition is ensured to be compatible with some DAGs. More details can be found in \cite{kuipers2017partition}.


\section{A New Strategy and Its Theoretical Properties} \label{sec:theory}

In Section \ref{subsec:strategy}, we present an effective strategy to handle hybrid Bayesian networks. Then it is shown that its performance can be evaluated in a two-node Bayesian network. The evaluation metrics for continuous BNs and discrete BNs are presented in Section \ref{sec:posteriorratio.bge} and \ref{sec:posteriorratio.bde} respectively. And finally in Section \ref{sec:compare.bge.bde}, theoretical results for the proposed strategy have been derived in all the four possible scenarios with two-node Bayesian networks.

\subsection{A Strategy to Handle HBNs} \label{subsec:strategy}
We propose a {\it Run it as Gaussian} (RAG) strategy which considers all random variables to be continuous (and Gaussian). Although RAG mis-specifies the data distribution, both our theoretical and empirical results show that RAG outperforms the discretization strategy.

When theoretically comparing the performance of several methods for structure learning, the first question needs to be answered is: how many nodes and edges should be included in the Bayesian networks? Our theoretical results aim to be general enough to provide a guidance to a large range of sizes of Bayesian networks. 
Fortunately, Lemma \ref{lemma:2nodes} shows that any edge in any Bayesian network can be represented as a dependency between two random variables. One immediate implication of this Lemma is that when evaluating the impact of any strategy (including RAG), we only need to investigate its impact or performance for two-node BN cases.

\begin{lemma} \label{lemma:2nodes}
Given $n$ random variables $\bX$, for any graph structure $\cG$, the posterior distribution $P(\cG \mid \bX)$ can be derived by conducting a series of  simple linear regressions.
\end{lemma}

\begin{proof}
A proof of Lemma \ref{lemma:2nodes} can be found in Appendix A. 
\end{proof}


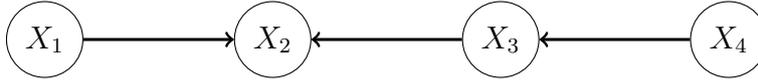
\begin{figure}[tb]
    \centering
    \begin{tikzpicture}
    \node[shape=circle,draw=black] (A) at (0,0) {$X_1$};
    \node[shape=circle,draw=black] (B) at (3,0) {$X_2$};
    \node[shape=circle,draw=black] (C) at (6,0) {$X_3$};
    \node[shape=circle,draw=black] (D) at (9,0) {$X_4$};
    
    \draw[->,line width=0.4mm] (A) -- (B);
    \draw[->,line width=0.4mm] (C) -- (B);
    \draw[->,line width=0.4mm] (D) -- (C);
        
    \end{tikzpicture}
    \caption{A DAG structure $\cG$.}
    \label{fig:example.BN}
\end{figure}

Lemma \ref{lemma:2nodes} shows that, essentially the posterior distribution involves a series of simple linear regressions, each of which can be deemed as a two-node Bayesian networks.
Once the impact of the strategies on two-node Bayesian networks is well understood, it is clear about the impact on any single edge in any Bayesian network. Thus, it provides a valuable indicator to the impact of the strategy on the whole Bayesian network. For example, suppose a strategy inclines to underestimate the existence of a link in a two-node Bayesian network, then such a strategy will lead to high false negative. However, the severity of false negative depends on the true graph.

To this end, we only need to examine the impact of strategies to two-node Bayesian networks. We consider $\bX=(X_1,X_2)$ and note that there are only two equivalence classes, either $X_1$ and $X_2$ are dependent so that an edge exists, or they are independent so that no edge exists, denoted by $\mathcal{G}_0$. 
In the following section, let $\cG_1$ denote graph $X_1 \rightarrow X_2$, which is set to be the ground truth graph. Naturally, $\cG_0$ is a special case of $\cG_1$ with the dependence equal to 0.

The metric that we use to evaluate the performance of each strategy is 
the expectation of log posterior ratio, defined as below
\begin{eqnarray} \label{eqn:evaluation.criterion}
&&r_{10} = \frac{1}{N} \int_{\bx \in \Omega_{\bX}}  \log\left(\frac{P(\cG_1 \mid \bX) }{P(\cG_0 \mid \bX )}\right) P(\bX = \bx) d\bx,
\end{eqnarray}
where $\Omega_{\bX}$ denotes the sample space for the random variables $(X_1, X_2)$ with $N$ samples. Given the data is generated from $\cG_1$, a larger posterior likelihood ratio corresponds to better performance of a strategy, except when the data is generated from $\cG_0$, which is a special case of $\cG_1$ with no dependence between the two nodes.

Section \ref{sec:posteriorratio.bge} and \ref{sec:posteriorratio.bde} derive the expectation of posterior ratio by using Bayesian Gaussian likelihood equivalence (BGe) \citep{heckerman1995} and Bayesian Dirichlet likelihood equivalence (BDe) \citep{heckerman1995learning} respectively, regardless the data generation process.

\subsection{Posterior Ratio using BGe Score} \label{sec:posteriorratio.bge}
Under the BGe score, $X_1$ and $X_2$ are modeled as bi-variate Gaussian random variables with $\bmu$ as mean vector and $\bW$ as precision matrix. The prior distribution on $\bmu$ and $\bW$ is a Normal-Wishart distribution \citep{geiger2002parameter,kuipers2014addendum}, which is defined as follows. The prior on $\bW$ is a Wishart distribution, $\bW \sim \mathcal{W}_n (\bT^{-1},\alpha_w)$, where $\alpha_w > n+1$ is the degrees of freedom and $\bT$ is the positive definite parametric matrix. More specifically,
\begin{eqnarray}
P(\bW \mid \bT, \alpha_w) = \frac{\mid \bW \mid^{\frac{\alpha_w -n -1}{2}} }{Z_{\mathcal{W}(n,\bT,\alpha_w)}} \exp\left\{-\frac{1}{2} Tr(\bT\bW)\right\}, \label{eqn:prior.W}
\end{eqnarray}
where the normalizing constant $Z_{\mathcal{W}(n,\bT,\alpha_w)}= \frac{2^{\frac{\alpha_w n}{2}} \Gamma_n(\frac{\alpha_w}{2})}{\mid \bT\mid^{\frac{\alpha_w}{2}}}$ and the multivariate gamma function is defined as $\Gamma_n(\frac{\alpha_w}{2}) = \pi^{n(n-1)/4} \prod_{j=1}^n \Gamma(\frac{\alpha_w+1-j}{2})$.
Given the precision matrix $\bW$, the prior distribution on $\bmu$ is a normal distribution, which is given as,
\begin{eqnarray}
P(\bmu\mid \bW,\alpha_\mu) = \frac{(\alpha_\mu)^{\frac{n}{2}} \mid \bW \mid ^{\frac{1}{2}}}{(2\pi)^{\frac{n}{2}}} \exp\left\{-\frac{1}{2} Tr\left((\bmu - \bnu)(\bmu - \bnu)^T \alpha_{\mu}\bW\right) \right\}. \label{eqn:prior.mu}
\end{eqnarray}
 
The hyperparameters in these prior distributions are $\alpha_\mu$, 
$\bnu=(0,0)^T$ and $\bT=t \bI_n$, where 
$\bI_n$ is the $n\times n$ identity matrix.
Lemma \ref{lemma:BGe.posterior.ratio} shows the computation of $r_{10}$ under the BGe score, where the variables are modeled as a multivariate Gaussian distribution regardless the true data generation processes.

\begin{lemma} \label{lemma:BGe.posterior.ratio}
Under the BGe setting, using the prior distributions shown in Equation \eqref{eqn:prior.W} and \eqref{eqn:prior.mu}, expectation of the posterior ratio of different graph structures given $\bX$ is computed as,
\begin{eqnarray}
\lim_{N\rightarrow \infty} r_{10} &=& \frac{1}{2} \log\left( \frac{\Sigma_{11} \Sigma_{22}}{\Sigma_{11} \Sigma_{22} - \Sigma_{12}^2} \right), \label{eqn:BGe.infinite.ratio}
\end{eqnarray}
where 
$\Sigma$ is the covariance matrix of $(X_1,X_2)$, and $\Sigma_{ij}$ denotes the element of the covariance matrix at row $i$ and column $j$.
\end{lemma}
The proof can be found in Appendix B. 
Obviously, for a large value of $N$, $r_{10}$ is mainly determined by covariance matrix $\Sigma$.

\subsection{Posterior Ratio Using BDe Score} \label{sec:posteriorratio.bde}
Under the BDe score, variables $X_1$ and $X_2$ are discrete variables, assuming binary variables for simplicity purpose. As there are two levels for both variables $X_1$ and $X_2$. There are a total of three parameters which can fully specify the joint distribution of $X_1$ and $X_2$: $\theta_{12}$ denoting the probability of $X_1=1, X_2=1$, $\theta_{1\bar{2}}$ denoting the probability of $X_1=1, X_2=0$ and $\theta_{\bar{1}2}$ denoting the probability of $X_1=0, X_2=1$. The prior distributions for the three parameter under Bayesian Dirichlet likelihood equivalence is a Dirichlet distribution, that is
\begin{eqnarray} \label{eqn:prior.BDe}
P(\theta_{12},\theta_{1\bar{2}},\theta_{\bar{1}2} \mid \alpha_1,\cdots,\alpha_4) = \frac{1}{\boldsymbol{B}(\balpha)} \theta_{12}^{\alpha_1 - 1} \theta_{1\bar{2}}^{\alpha_2 -1 } \theta_{\bar{1}2}^{\alpha_3 -1} (1-\theta_{12}-\theta_{1\bar{2}}-\theta_{\bar{1}\bar{2}})^{\alpha_4 -1},
\end{eqnarray}
where $\boldsymbol{B}(\balpha)$ is a multivariate beta function, which is expressed in terms of the gamma function $\boldsymbol{B}(\balpha) = \frac{\prod_{i=1}^K \Gamma(\alpha_i)}{\Gamma(\sum_{i=1}^K \alpha_i)}$. 
Lemma \ref{lemma:BDe.posterior.ratio} shows the computation of $r_{10}$ under the BDe setting, where the variables are modeled as a multivariate multinomial distribution.

\begin{lemma}
\label{lemma:BDe.posterior.ratio}
Suppose both $X_1$ and $X_2$ are binary variables, and the true distribution of $X_1,X_2$ is denoted by $P(X_1,X_2)$. Denoting the hyperparameters in the prior distribution as $\balpha = (\alpha_1,\alpha_2,\alpha_3,\alpha_4)$, the expectation of the posterior ratio given $\bX$ is computed as,
\begin{eqnarray}
\lim_{N \rightarrow \infty}r_{10} =  p_{11} \log\left(\frac{p_{11}}{p_{1.} p_{.1}}\right) + p_{10} \log\left(\frac{p_{10}}{  p_{1.} p_{.0} }\right) + p_{01} \log\left(\frac{ p_{01}}{ p_{0.} p_{.1}}\right) + p_{00} \log\left(\frac{ p_{00}}{p_{0.} p_{.0}}\right), 
\end{eqnarray}
where $N_{11}$ denotes the number of samples in $\bX$ equal to $(1,1)$, $N_{10}$ denotes the number of samples in $\bX$ equal to $(1,0)$, $N_{01}$ denotes the number of samples in $\bX$ equal to $(0,1)$, $N_{00}$ denotes the number of samples in $\bX$ equal to $(0,0)$, $N_{1.} = N_{11}+N_{10}$, $N_{0.} = N_{01}+N_{00}$, $p_{1.} = \frac{N_{11}+ N_{10}}{N} = p_{11}+p_{10}, p_{0.} = \frac{N_{01}+ N_{00}}{N} = p_{01}+p_{00}, p_{.1} = \frac{N_{11}+ N_{01}}{N} = p_{11}+p_{01}$, and $p_{.0} = \frac{N_{10}+ N_{00}}{N} = p_{10}+p_{00}$.
\end{lemma}
A proof can be found in Appendix C. 
For a large sample size $N$, $r_{10}$ is determined by the proportions of realization of a bi-variate random variable. A direct implication of of Lemma \ref{lemma:BDe.posterior.ratio} is that: if $X_1$ and $X_2$ is independent, then $\lim_{N \rightarrow \infty}r_{10} =0$, due to the fact that, in the independent case, $p_{11} = p_{1.} p_{.1}, p_{10} = p_{1.} p_{.0}, p_{01} = p_{0.} p_{.1}$ and $p_{00} = p_{0.} p_{.0}$.

\subsection{Comparison between RAG and Discretization} \label{sec:compare.bge.bde}

In this section, we theoretically compare the relative performance of RAG with discretization strategy only, due to the fact that in the literature discretization strategy is more prevalent than other strategies. Naturally, the strategy RAG indicates that BGe score is adopted as the score function. Similarly, the  discretization strategy indicates that BDe score is adopted as the score function. For the rest of the paper, we refer the discretization strategy to as {\bf DISC} for simplicity purpose. 

The discretized data from the original data is denoted by $\bX'$. Analogy to $r_{10}$, we define $\tilde{r}_{10}$ to measure the performance of BDe for data $\bX'$.
\begin{eqnarray} \label{eqn:evaluation.criterion.BDe}
&&\tilde{r}_{10} = \frac{1}{N} \sum_{\bx \in \Omega_{\bX'}}  \log\left(\frac{P(\cG_1 \mid \bX') }{P(\cG_0 \mid \bX' )}\right) P(\bX' = \bx),
\end{eqnarray}
where $\Omega_{\bX'}$ denotes the sample space for a discrete bi-variate random variable with $N$ samples.

\subsubsection{Data Generation Processes} \label{subsec:dgp}

As shown by Lemma \ref{lemma:2nodes}, it is adequate to study the performance of any strategy in a two-node Bayesian network. Considering each one of the two nodes could be either continuous or discrete random variable, there are four combinations of the Bayesian network. It is desirable to define the four possible scenarios and derive theoretical results for all scenarios.
The four different data generation scenarios are defined as follows,
\begin{enumerate}
\item \textit{Continuous to Continuous} ($\mathcal{S}_{cc}$).
We refer to the case where both $X_1$ and $X_2$ are continuous variables, as $\mathcal{S}_{cc}$. Let $X_1 \sim \cN(\mu_1, \sigma^2_1)$ and $X_2 \mid X_1=x_1 \sim \cN(\mu_2 + \beta x_1, \sigma^2_{2})$.
The marginal distribution of $X_2$ is $X_2 \sim \cN(\tilde{\mu}_2 , \tilde{\sigma}_2^2)$, where $\tilde{\mu}_2 = \mu_2+\beta \mu_1$ and $\tilde{\sigma}_2^2 =\beta^2 \sigma_1^2+\sigma^2_2$. Then the joint distribution of
$(X_1,X_2)$ is $\mathcal{MVN}\;(\boldsymbol{\mu}, \bW)$, where $\bmu$ is the mean parameter and $\bW$ is the precision matrix
\begin{align*}
\bmu=(\mu_1,\mu_2+\beta \mu_1)^T, \ 
\bW^{-1} =\left[ \begin{array}{cc}
    \sigma_1^2 & \beta \sigma_1^2  \\
     \beta \sigma_1^2 & \beta^2\sigma_1^2 + \sigma^2_2
\end{array}\right].
\end{align*}
\item \textit{Continuous to Discrete} ($\mathcal{S}_{cd}$).
We refer to the case where $X_1$ is continuous and $X_2$ is discrete as $\mathcal{S}_{cd}$. Let $X_1 \sim \cN(\mu_1, \sigma^2_1)$ and $X_2 \mid X_1=x_1 \sim \cB e(p)$, where $p= \displaystyle\frac{\exp\{\beta\times(x_1 - \mu_1)\}}{1 + \exp\{\beta \times(x_1 - \mu_1)\}}$ and $\cB e(p)$ denotes a Bernoulli random variable with success probability equal to $p$. 
\item  \textit{Discrete to Continuous} ($\mathcal{S}_{dc}$).
We refer to the case where $X_1$ is discrete and $X_2$ is continuous as $\mathcal{S}_{dc}$, where $X_1 \sim \mathcal{B}e(p)$ and $X_2 \mid X_1=x_1 \sim \cN(\mu_2+\beta x_1, \sigma^2_{2})$, 
so that $X_1$ is a Bernoulli random variable with probability $p$, and $X_2$ is a Gaussian random variable.
\item \textit{Discrete to Discrete} ($\mathcal{S}_{dd}$).
We refer to the case where both node $X_1$ and $X_2$ are discrete variables as $\mathcal{S}_{dd}$. Let $X_1 \sim \mathcal{B}e(p)$. Then $X_2 \mid X_1=1 \sim \mathcal{B}e(0.5 + \frac{\beta}{2})$ and $X_2 \mid X_1=0 \sim \mathcal{B}e(0.5 - \frac{\beta}{2})$, where $\beta \in [0,1]$.
\end{enumerate}
For hybrid Bayesian networks, different strategies are implemented, either RAG or discretization strategy. Section \ref{subsec:theory} compares the performance of different strategies under four data generation processes.


\subsubsection{Theoretical Results} \label{subsec:theory}

Theorem \ref{thm:ratio.scc} demonstrates the expectation of the posterior ratios derived from strategy RAG (and BGe score function) and DISC (and BDe score function) under the scenario $\cS_{cc}$ respectively.
\begin{theorem} [Posterior ratio under $\cS_{cc}$] \label{thm:ratio.scc}
Suppose the data $\bX$, with $N$ samples, is generated from the two continuous variables case $\mathcal{S}_{cc}$ (Section \ref{subsec:dgp}). The expectation of the posterior ratio using RAG is,
\begin{eqnarray}
\lim_{N\rightarrow \infty} r_{10} = \frac{1}{2} \log\left( \frac{\sigma_2^2 + \beta^2 \sigma_1^2}{\sigma_2^2} \right). \label{eqn:bge.scc.app}
\end{eqnarray}

Let $\bX^{'}$ denote the data $\bX$ after discretization (using median as cutting point). The expectation of the posterior ratio using DISC is,
\begin{eqnarray} \label{eqn:upper.bde.scc}
\lim_{N \rightarrow \infty}\tilde{r}_{10} =\log(4) + 2 \tilde{p}_{11} \log(\tilde{p}_{11}) + (1-2\tilde{p}_{11}) \log\left(\frac{1}{2}-\tilde{p}_{11}\right),
\end{eqnarray}
where $\tilde{p}_{11} =\int_{0}^{\infty} \frac{1}{\sqrt{2\pi } } \exp\{- \frac{x^2}{2 } \} \Phi(\frac{\sigma_1\beta x}{\sigma_2})dx $ and $\Phi(\cdot)$ is the cumulative density function of the standard normal distribution.
\end{theorem}
A proof of Theorem~\ref{thm:ratio.scc} can be found in Appendix D. 
Given the data is generated under $\cS_{cc}$, Figure \ref{post.ratio.scc} demonstrates the $r_{10}$ and $\tilde{r}_{10}$ as $\beta$ increases and fixing $\sigma_1=\sigma_2=1$. We can clearly see that our strategy RAG using BGe outperforms the strategy DISC. Such a conclusion is not surprising, due to the information loss when using DISC. It is noticeable that $\tilde{r}_{10}$ has an upper limit equal to $\log\left( 2\right)$, as you can see the upper limit of $\tilde{p}_{11}$ in Equation \eqref{eqn:upper.bde.scc} is $\frac{1}{2}$. This implies that BDe score function may be unable to show extremely strong dependency between two variables. By contrast, BGe score can effectively capture the signal strength when $\beta$ increases as seen from Equation \eqref{eqn:bge.scc.app}.

\begin{figure}[tb]
    \centering
    \begin{minipage}{0.47\textwidth}
        \centering
        \includegraphics[width=\textwidth]{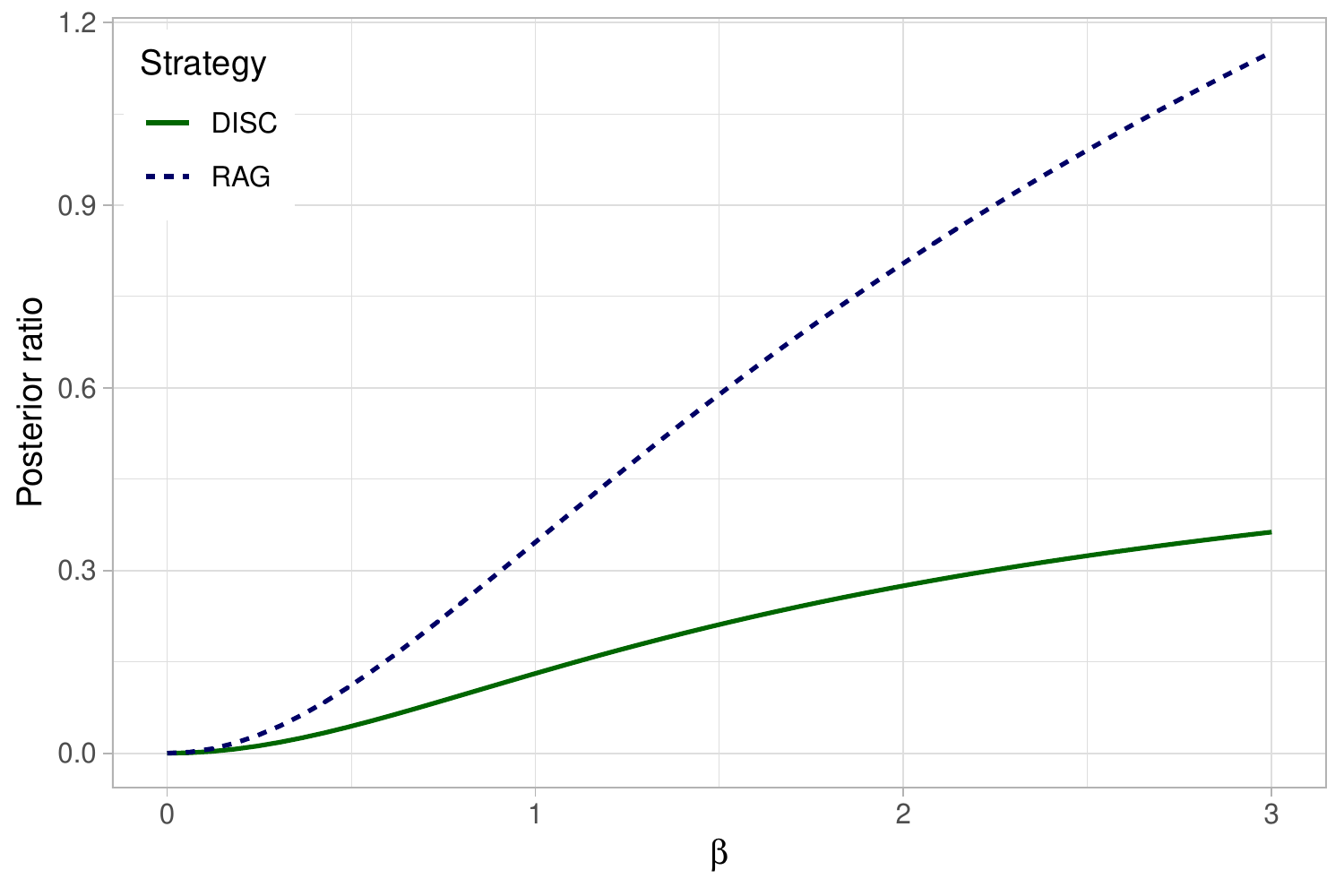} 
        \subcaption{$\cS_{cc}$} \label{post.ratio.scc}
    \end{minipage}
    \begin{minipage}{0.47\textwidth}
        \centering
        \includegraphics[width=\textwidth]{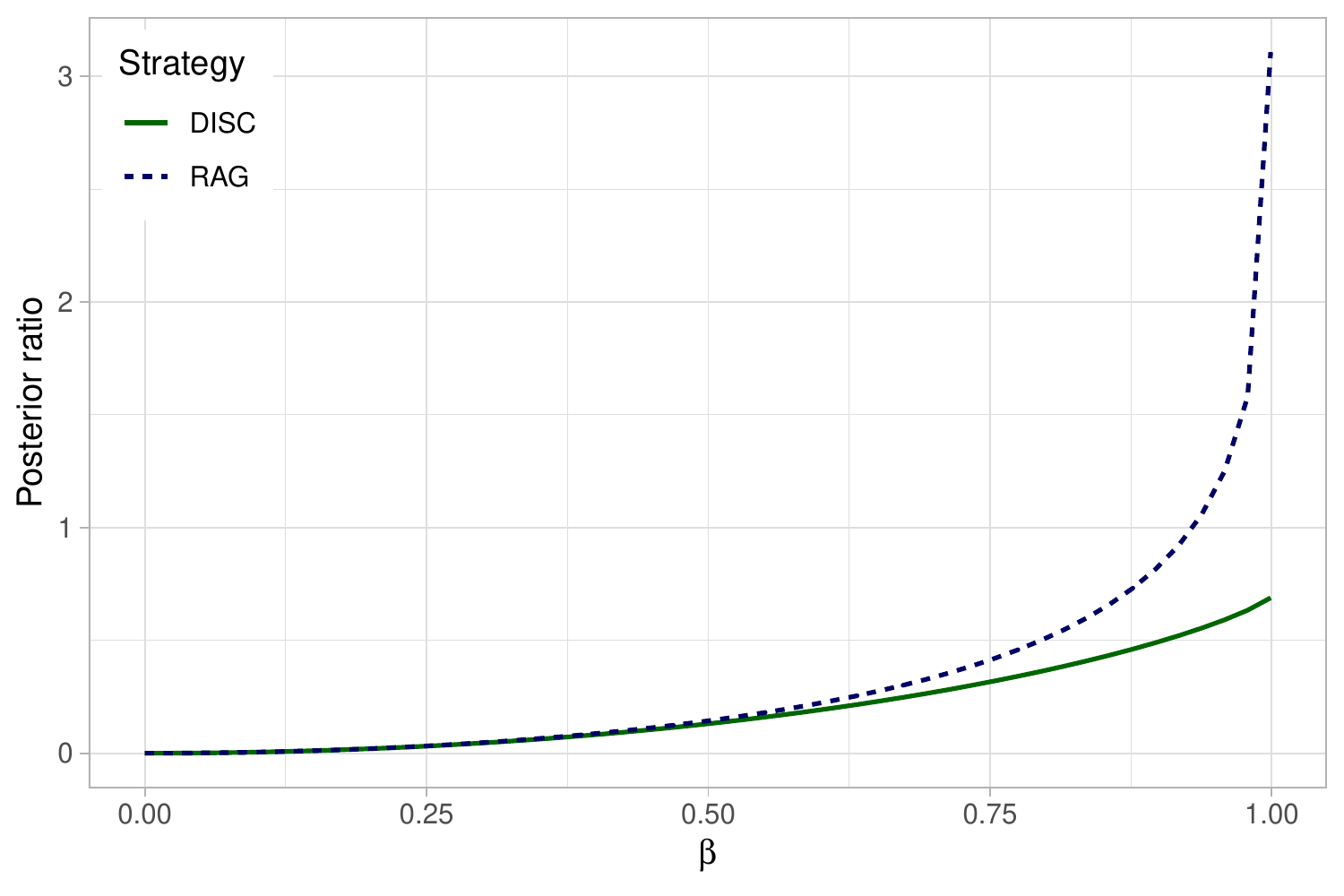} 
        \subcaption{$\cS_{dd}$} \label{post.ratio.sdd}
    \end{minipage}
    
    \begin{minipage}{0.47\textwidth}
        \centering
        \includegraphics[width=\textwidth]{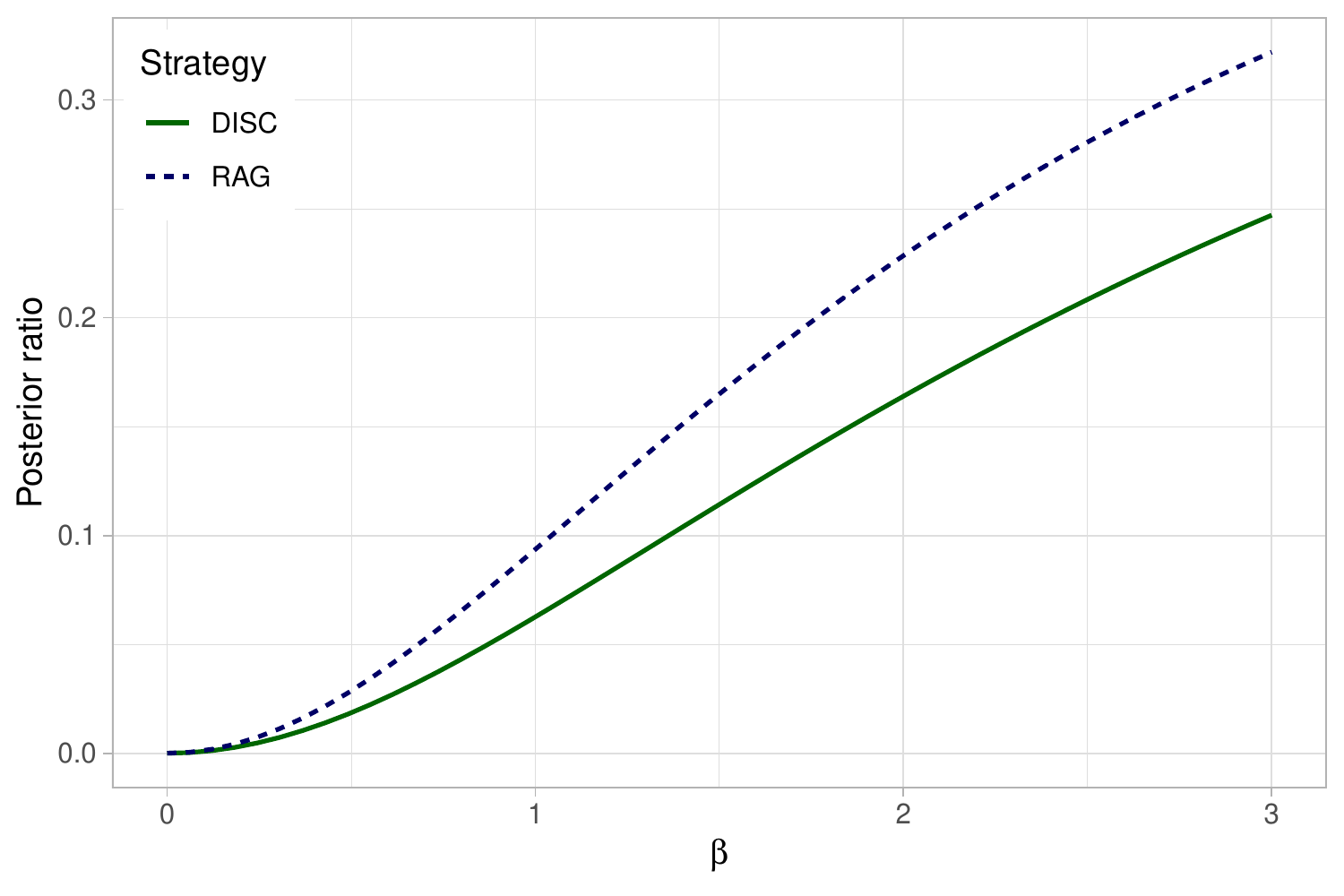} 
        \subcaption{$\cS_{cd}$} \label{post.ratio.scd}
    \end{minipage}
    \begin{minipage}{0.47\textwidth}
        \centering
        \includegraphics[width=\textwidth]{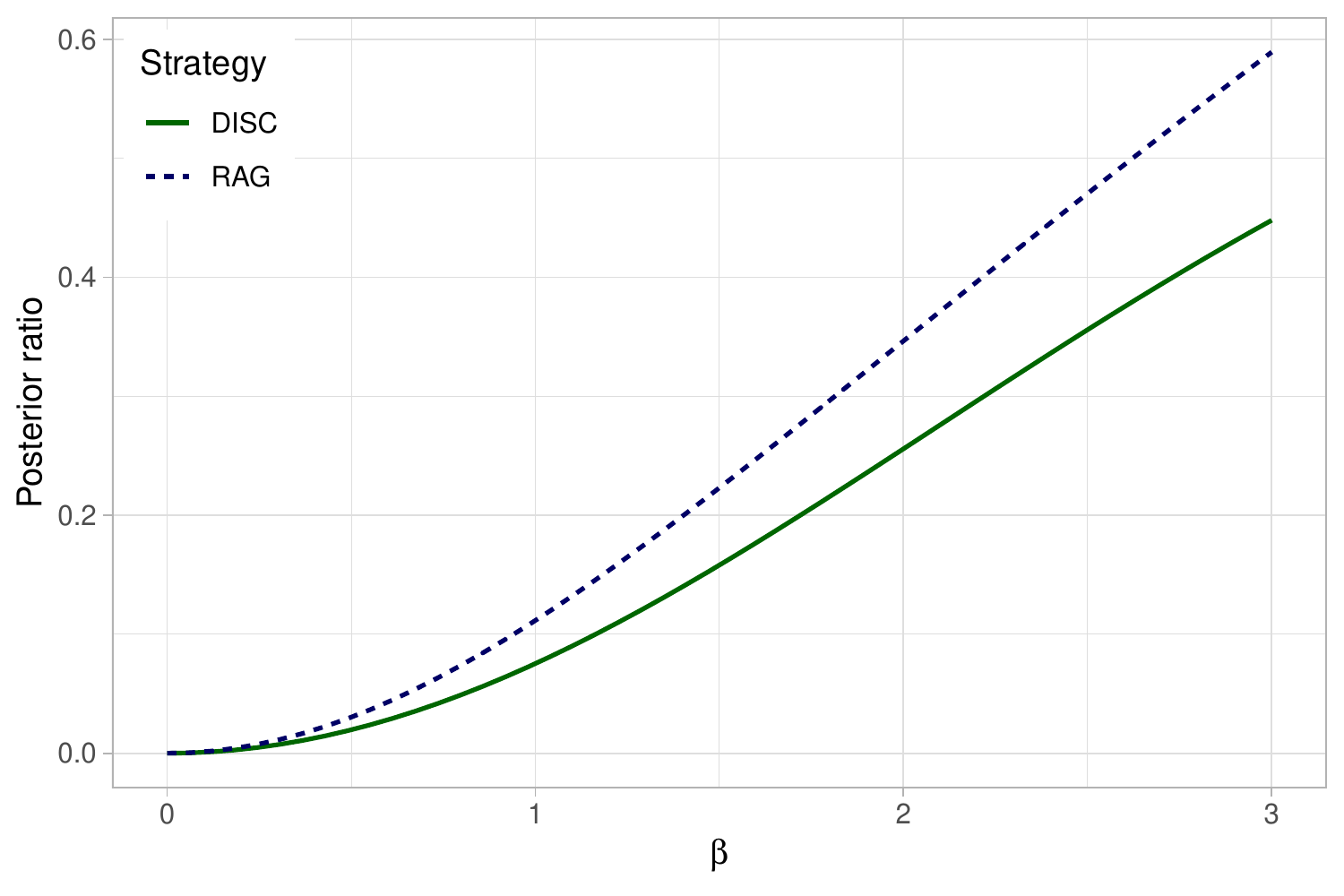} 
        \subcaption{$\cS_{dc}$} \label{post.ratio.sdc}
    \end{minipage}
       \caption{The approximation of posterior ratios $r_{10}$ and $\tilde{r}_{10}$ as $\beta$ increases under the scenarios of \protect\subref{post.ratio.scc} $\mathcal{S}_{cc}$ (continuous to continuous) with $\sigma_1=\sigma_2=1$; \protect\subref{post.ratio.sdd} $\mathcal{S}_{dd}$ (discrete to discrete) as $p = 0.5$ is fixed; \protect\subref{post.ratio.scd} $\mathcal{S}_{cd}$ (continuous to discrete) as $\sigma_1=1$ is fixed; \protect\subref{post.ratio.sdc} $\mathcal{S}_{dc}$ (discrete to continuous) as $\sigma_2=1$ is fixed.}
\end{figure}

Theorem \ref{thm:ratio.scd} demonstrates the expectation of the posterior ratios derived from the strategy RAG (and BGe score function) and DISC (and BDe score function) setting under the scenario $\cS_{cd}$ respectively.
\begin{theorem}  [Posterior ratio under $\cS_{cd}$] \label{thm:ratio.scd}
Suppose the data $\bX$, with $N$ samples, is generated from the case $\mathcal{S}_{cd}$ (Section \ref{subsec:dgp}), where a continuous variable is the parent of a discrete variable. Then, the expectation of the posterior ratio using RAG is,
\begin{eqnarray}
\lim_{N \rightarrow \infty} r_{10}&=& \frac{1}{2} \log \left( \frac{\Sigma_{11} \Sigma_{22}}{ \Sigma_{11} \Sigma_{22} - \Sigma_{12}^2} \right), \label{eqn:bge.scd.app}
\end{eqnarray}
where $\Sigma_{11} = \sigma_1^2$, $\Sigma_{12} = \int_{-\infty}^{+\infty} \frac{t \sigma_1}{\sqrt{2\pi}} \frac{\exp(\beta \sigma_1 t)}{1+ \exp(\beta \sigma_1 t)}  \exp\left\{ -\frac{t^2}{2}\right\}dt$, $\Sigma_{22} = \bbE(X_2) - (\bbE(X_2))^2$ and $\bbE(X_2)= \int_{-\infty}^{+\infty} \frac{1}{\sqrt{2\pi}} \frac{\exp(\beta \sigma_1 t)}{1+ \exp(\beta \sigma_1 t)}  \exp\left\{ -\frac{t^2}{2}\right\}dt$.

Let $\bX^{'}$ denote the data $\bX$ after discretization (using median as cutting point). The expectation of the posterior ratio using DISC is,
\begin{eqnarray} \label{eqn:upper.bde.scd}
\lim_{N \rightarrow \infty}\tilde{r}_{10} =\log(4) + 2 \tilde{p}_{11} \log(\tilde{p}_{11}) + (1-2\tilde{p}_{11}) \log\left(\frac{1}{2}-\tilde{p}_{11}\right),
\end{eqnarray}
where $\tilde{p}_{11} =\int_{0}^{\infty} \frac{1}{\sqrt{2\pi } } \exp\left\{- \frac{x^2}{2 } \right\} \frac{\exp\{\beta x \sigma_1 \}}{1+\exp\{\beta x \sigma_1 \}}dx$.
\end{theorem}

\begin{proof}
A proof can be found in Appendix E. 
\end{proof}
Given the data is generated under $\mathcal{S}_{cd}$, Figure \ref{post.ratio.scd} demonstrates the $r_{10}$ and $\tilde{r}_{10}$ as $\beta$ increases and fixing $\sigma_1=1$. We can see clearly that our strategy RAG using BGe outperforms the DISC and BDe score function. Similar with the scenario $\cS_{cc}$, Equation \eqref{eqn:upper.bde.scd} indicates that the posterior ratio using DISC has an upper limit as $\log(2)$, since $\tilde{p}_{11} \rightarrow 0.5$ as $\beta \rightarrow \infty$.

Theorem \ref{thm:ratio.sdc} demonstrates the expectation of the posterior ratios derived from the strategy RAG (and BGe score function) and DISC (and BDe score function) under the scenario $\cS_{dc}$ respectively.
\begin{theorem} [Posterior ratio under $\cS_{dc}$] \label{thm:ratio.sdc}
Suppose the data $\bX$, with $N$ samples, is generated from the case $\mathcal{S}_{dc}$ (Section \ref{subsec:dgp}) with $p=\frac{1}{2}$, where a discrete variable is the parent of a continuous variable. Then, the expectation of the posterior ratio using RAG is,
\begin{eqnarray}
\lim_{N\rightarrow \infty}r_{10} = \frac{1}{2} \log \left( 1+ \frac{\beta^2}{4 \sigma_2^2}\right). \label{eqn:bge.sdc.app} 
\end{eqnarray}

Let $\bX^{'}$ denote the data $\bX$ after discretization (using median as cutting point). The expectation of the posterior ratio using DISC is,
\begin{eqnarray}
 \lim_{N \rightarrow \infty}\tilde{r}_{10} = \log(4) + 2 \tilde{p}_{11} \log(\tilde{p}_{11}) + (1-2\tilde{p}_{11}) \log\left(\frac{1}{2}-\tilde{p}_{11}\right),
\end{eqnarray}
where $\tilde{p}_{11} =\frac{1}{2} \int_{-\frac{\beta}{2 \sigma_2}}^{\infty} \frac{1}{\sqrt{2\pi } } \exp\{- \frac{x^2}{2 } \}dx $. 
\end{theorem}

\begin{proof}
A proof can be found in Appendix F. 
\end{proof}

Given the data is generated under $\mathcal{S}_{dc}$, Figure \ref{post.ratio.sdc} demonstrates the $r_{10}$ and $\tilde{r}_{10}$ as $\beta$ increases and fixing $\sigma_2=1$. We can see clearly that our strategy RAG using BGe outperforms the DISC strategy. Similarly, $\tilde{r}_{10}$ has an upper limit $\log(2)$ since $\tilde{p}_{11} \rightarrow 0.5$ as $\beta$ increases. While RAG can always capture the dependency as $\beta$ changes, as shown in Equation \eqref{eqn:bge.sdc.app}.

Theorem \ref{thm:ratio.sdd} demonstrates the expectation of the posterior ratios derived from the strategy RAG (and BGe score function) and DISC (and BDe score function) under the scenario $\cS_{dd}$ respectively.

\begin{theorem} [Posterior ratio under $\cS_{dd}$] \label{thm:ratio.sdd}
Suppose the data $\bX$, with $N$ samples, is generated from the case $\mathcal{S}_{dd}$ (Section \ref{subsec:dgp}), where a discrete variable is the parent of another discrete variable. Then, the expectation of the posterior ratio using RAG is,
\begin{eqnarray}
\lim_{N \rightarrow \infty}r_{10} =  \frac{1}{2} \log\left( \frac{1- (2p-1)^2 \beta^2}{1- \beta^2} \right). \label{eqn:bge.sdd.app}
\end{eqnarray}
The expectation of the posterior ratio using DISC is,
\begin{eqnarray*}
\lim_{N \rightarrow \infty}\tilde{r}_{10} & = & p_{11} \log\left(\frac{p_{11}}{p_{1.} p_{.1}}\right) + p_{10} \log\left(\frac{p_{10}}{  p_{1.} p_{.0} }\right) + p_{01} \log\left(\frac{ p_{01}}{ p_{0.} p_{.1}}\right) + p_{00} \log\left(\frac{ p_{00}}{p_{0.} p_{.0}}\right),
\end{eqnarray*}
where $p_{11} = p (0.5+\frac{\beta}{2}), p_{10} = p (0.5-\frac{\beta}{2}), p_{01} = (1-p) (0.5-\frac{\beta}{2})$, $p_{00} = (1-p) (0.5+\frac{\beta}{2})$, $p_{1.} = p_{11}+ p_{10}$, $p_{0.} = p_{01}+ p_{00}$, $p_{.1} = p_{11}+ p_{01}$ and $p_{.0} = p_{10}+ p_{00}$.
\end{theorem}
\begin{proof}
A proof can be found in Appendix G. 
\end{proof}

Given the data is generated under $\mathcal{S}_{dd}$, Figure \ref{post.ratio.sdd} demonstrates the $r_{10}$ and $\tilde{r}_{10}$ as $\beta$ increases and fixing $p=0.5$. We can see clearly that our strategy RAG using BGe outperforms the strategy DISC using BDe score function.

It is surprising to see that the RAG outperforms DISC under $\mathcal{S}_{dd}$. As $\beta \rightarrow 1$, $r_{10} \rightarrow \infty$, while $\tilde{r}_{10}$ will hit its upper limit $ \log(2)$.
It is anti-intuitive at the first glance that RAG is better than DISC under scenario $\cS_{dd}$. The explanation is that BGe score is sensitive to the changes of $\beta$, especially when $\beta$ is large. In an extreme case, as $\beta \rightarrow 1$ in the case of $\cS_{dd}$, the variance of the conditional Gaussian distribution will be close to 0, and the binary data is extremely unbalanced, i.e., $x_1 = 1$ leads to $x_2 =1$ with probability $1$. The posterior ratio of such data using BGe score function will approach $\infty$ rapidly. This phenomenon indicates that BGe score becomes beneficial when the sample size is small and/or $\beta$ is relatively small. 

In summary, for all the four scenarios, DISC results in information loss, leading to a metric $\tilde{r}_{10}$ with an upper bound. On contrary, RAG strategy does not suffer from such a constraint except the scenario $\cS_{cd}$. RAG can effectively capture the change of dependency between the two nodes. Thus, RAG is preferable to DISC strategy.

\section{Simulation Study} \label{sec:simulations}
\subsection{Simulation Settings}

\subsubsection{Settings of BN with 2 Nodes} \label{subsubsec:sim.2nodes}
A total of 100 data sets were generated for each of the four scenarios ($\cS_{cc}$, $\cS_{cd}$, $\cS_{dc}$, and $\cS_{dd}$) for simulation study. Each data set contains 2 variables $A$ and $B$, with 200 observations for each variable. For all the scenarios, the true DAG is $A \rightarrow B$. The parameter values used in each data generation process are given as follows,
\begin{itemize}
    \item [$\cS_{cc}$:] Both nodes $A$ and $B$ are continuous random variables, where $A \sim \mathcal{N}(\mu_A = -1, \sigma_A = 1)$ and $B \mid A = a \sim \mathcal{N}(\beta a, 1)$.
    \item [$\cS_{cd}$:] Node $A$ is a continuous random variable, $A \sim \mathcal{N}(\mu_A = -1, \sigma_A = 1)$, while node $B$ is a Bernoulli random variable with probability $p = \displaystyle\frac{\exp\big(\beta (a - \mu_A)\big)}{1 + \exp\big(\beta(a - \mu_A)\big)}.$
    \item [$\cS_{dc}$:] Node $A$ is a Bernoulli random variable with probability $p = 0.5$, while $B$ is a continuous random variable, $B \mid A = a \sim \mathcal{N}(\beta a, 1)$.
    \item [$\cS_{dd}$:] Both nodes $A$ and $B$ are discrete random variables, where node $A$ is a Bernoulli random variable with probability $p = 0.1$. Node $B$ is defined as follows: $B \mid A = 0 \sim \cB e( 0.5 - \frac{\beta}{2})$ and $B \mid A = 1 \sim \cB e( \frac{\beta}{2} + 0.5)$.
\end{itemize}

\subsubsection{Settings of BN with 4 Nodes}
To understand the performance of RAG in a broader scope, simulation studies were also conducted for Bayesian networks with 4 nodes, including both continuous and discrete variables. The true DAG structure is visualized in Figure~\ref{fig:dag4nodes}.
\begin{figure}[tb]
    \centering
    \begin{tikzpicture}
    \node[shape=circle,draw=black,fill=RoyalBlue!25] (A) at (3,2) {$A$};
    \node[shape=circle,draw=black,fill=ForestGreen!25] (B) at (0,0) {$B$};
    \node[shape=circle,draw=black,fill=RoyalBlue!25] (C) at (6,2) {$C$};
    \node[shape=circle,draw=black,fill=ForestGreen!25] (D) at (9,0) {$D$};
    
    \draw[->,line width=0.25mm] (A) -- (B);
    \draw[->,line width=0.25mm] (C) -- (B);
    \draw[->,line width=0.25mm] (C) -- (D);
    \draw[->,line width=0.25mm] (A) -- (D);
    \end{tikzpicture}
    \caption{A toy example of Bayesian network.}
    \label{fig:dag4nodes}
\end{figure}
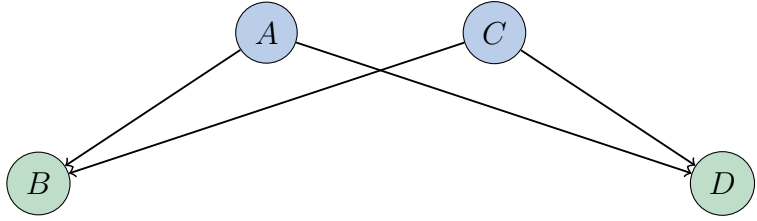

In total, 100 data sets were generated for each scenario. Each data set contains four variables, labelled $A$ to $D$, with 200 observations for each variable. The parameter values used in each data generation process are given as follows,
\begin{itemize}
    \item [$\cS_{cc}$:] All nodes are continuous random variables, where $A \sim \mathcal{N}(\mu_A = -3, \sigma_A = 1)$ and $C \sim \mathcal{N}(\mu_C = 6, \sigma_C = 1)$ are parent nodes. $B \mid A = a, C = c \sim \mathcal{N}(1.5a + 3c, 1)$ and $D \mid A = a, C = c \sim \mathcal{N}(2a + 1.5c, 1)$.
    \item [$\cS_{cd}$:] Parent nodes $A$ and $C$ are continuous random variables, where $A \sim \mathcal{N}(\mu_A = -3, \sigma_A = 1)$ and $C \sim \mathcal{N}(\mu_C = 6, \sigma_C = 1)$. Both $B$ and $D$ are Bernoulli random variables with probabilities:
    $p_B = \displaystyle\frac{\exp\big\{2(a + c - \mu_A - \mu_C)\big\}}{1 + \exp\big\{2(a + c - \mu_A - \mu_C)\big\}}$ and $p_D = \displaystyle\frac{\exp\big\{-1.5(a + c - \mu_A - \mu_C)\big\}}{1 + \exp\big\{-1.5(a + c - \mu_A - \mu_C)\big\}}$.
    \item [$\cS_{dc}$:] Nodes $A$ and $C$ are discrete random variables with $A$ having 2 labels $``1", ``2"$, $C$ having 4 labels $``1", ``2", ``3", ``4"$ simulated randomly. Nodes $B$ and $D$ are continuous random variables with $B \mid A = a, C = c \sim \mathcal{N}(1.5a + 3c, 1)$ and $D \mid A = a, C = c \sim \mathcal{N}(2a + 1.5c, 1)$.
    \item [$\cS_{dd}$:] All nodes are discrete random variables, where parent nodes $A \sim \cB e(0.5)$ and $C$ randomly sampled with replacement from $\{1, 2, 3, 4\}$. Both $B$ and $D$ are Bernoulli random variables with the following probabilities shown in Table \ref{tab:cpd}. 
\end{itemize}

\begin{table}[tb]
        \centering
        \resizebox{\textwidth}{!}{%
            \begin{tabular}{ccccc} 
            \toprule
             $C = c$ & $p(B=1|A=1,C=c)$ & $p(D=1|A=1,C=c)$ & $p(B=1|A=2,C=c)$ & $p(D=1|A=2,C=c)$ \\
             \midrule
             1 & 0.05 & 0.95 & 0.10 & 0.90\\
             2 & 0.10 & 0.90 & 0.30 & 0.70 \\ 
             3 & 0.30 & 0.70 & 0.70 & 0.30 \\
             4 & 0.70 & 0.30 & 0.95 & 0.05 \\
             \bottomrule
            \end{tabular}
            }
            \caption{The conditional distributions of $B$ and $D$ given $A$ and $C$.}
            \label{tab:cpd}
        \end{table}

\subsection{Methods and Evaluation Criteria}

Partition MCMC was used for structure learning using the publicly available \textbf{R} package \emph{BiDAG}, with 100,000 iterations implemented on each of these data sets. The two scores in the package are ``{\it bge}" and ``{\it bdecat}", which correspond to BGe and BDe score functions respectively. In terms of the hyperparameters, we adopt the default values in {\it BiDAG} package, i.e., $\alpha_\mu = 1, \alpha_w = n + \alpha_\mu + 1,  t = \alpha_\mu\times(\alpha_w - n - 1)/(\alpha_\mu + 1)$ for ``{\it bge}" score.

Two strategies are compared through the simulation: RAG and discretization (DISC) strategy. RAG considers all the variables in a Bayesian network to be continuous and run partition MCMC with BGe score. Whenever DISC is used, if necessary, to transform a continuous variable to be discrete, the equal quantile rule applies. Let DISC-$q$ denotes the discretizing a continuous variable into a $q$-level discrete variable. For example, DISC-$4$ indicates that a continuous variable is discretized into a 4-level discrete variable by adopting its 25\%, 50\%, 75\% quantiles as cutting points.

The well known methods, such as the conditional linear Gaussian (CLG) model, is not included in the simulation. This is due to the fact that, by definition, CLG model does not allow continuous nodes to have discrete nodes.
The following performance criteria are adopted.
\begin{itemize}
    \item True positive (TP) refers to the number of edges in the skeleton of the estimated graph which are also present in the skeleton of the true graph. The higher TP is the better performance is.
    \item False positive (FP) refers to the number of edges in the skeleton of the estimated graph which are absent in the skeleton of the true graph.
    \item False negative (FN) is the number of edges in the skeleton of the true graph which are absent in the skeleton of the estimated graph.
    \item Structural Hamming Distance (SHD) between two graphs is computed as $\text{SHD} = \text{FN} + \text{FP} + \#\text{\{edges with an error in direction\}}$. Lower SHD corresponds to better performance.
    \item True positive rate (TPR) is given by $\frac{\text{TP}}{\text{TP} + \text{FN}}$;
    \item Frequency ratio (FR) is computed as the ratio between $\cG_1$ (including its equivalence class) and $\cG_0$ in the posterior samples. Suppose $c_1^{(1)}, c_1^{(2)}, \cdots, c_1^{(T)}$ denote the number of $\cG_1$ in all the data replications. Similarly, $c_0^{(1)}, c_0^{(2)}, \cdots, c_0^{(T)}$ denote the number of $\cG_0$ in all the data replications. Mathematically, FR = $\frac{\sum_{i=1}^T c_1^{(i)} }{\sum_{i=1}^T c_0^{(i)}}$. The higher FR is, the more capable the approach is.
\end{itemize}
\subsection{Simulation Results}

\subsubsection{Results of BN with 2 Nodes}

The performances of RAG and DISC-2 strategies under the four considered scenarios are recorded in Table~\ref{tab:2nodes.results}. For $\cS_{cc}$, where both of the variables are continuous, it can be clearly observed that the proposed strategy RAG using BGe score always outperform DISC in terms of all the performance criteria. This observation is expected since both variables in the simulated data sets are generated from continuous distributions and certain information will be lost during the discretization process when BDe score function is used.

\begin{table*}[tb]\centering
\resizebox{\textwidth}{!}{%
\begin{tabular}{ccrrcrrcrrcrrcrrcrr}
\toprule
& & \multicolumn{2}{c}{$\beta = 0.05$} & \phantom{abc}& \multicolumn{2}{c}{$\beta = 0.1$} &
\phantom{abc} & \multicolumn{2}{c}{$\beta = 0.5$} &
\phantom{abc} & \multicolumn{2}{c}{$\beta = 1$} &
\phantom{abc} & \multicolumn{2}{c}{$\beta = 1.5$} &
\phantom{abc} & \multicolumn{2}{c}{$\beta = 2$}\\
\cmidrule{3-4} \cmidrule{6-7} \cmidrule{9-10} \cmidrule{12-13} \cmidrule{15-16} \cmidrule{18-19}
Scenario & Metric & RAG & DISC-2 && RAG & DISC-2 && RAG & DISC-2 && RAG & DISC-2 && RAG & DISC-2 && RAG & DISC-2 \\\midrule
\multirow{6}{*}{$\cS_{cc}$} & SHD & 0.99 & 0.99 & & 0.97 & 0.97 & & 0.47 & 0.59 & & 0.50 & 0.50 & & 0.50 & 0.51 & & 0.54 & 0.52 \\
& TP & 0.03 & 0.01 & & 0.16 & 0.06 & & 1.00 & 0.86 & & 1.00 & 1.00 & & 1.00 & 1.00 & & 1.00 & 1.00\\
& FP & 0.00 & 0.00 & & 0.00 & 0.00 & & 0.00 & 0.00 & & 0.00 & 0.00 & & 0.00 & 0.00 & & 0.00 & 0.00 \\
& FN & 0.97 & 0.99 & & 0.84 & 0.94 & & 0.00 & 0.14 & & 0.00 & 0.00 & & 0.00 & 0.00 & & 0.00 & 0.00 \\
& TPR & 0.00 & 0.01 & & 0.16 & 0.06 & & 1.00 & 0.86 & & 1.00 & 1.00 & & 1.00 & 1.00 & & 1.00 & 1.00\\
& FR & 0.23 & 0.11 & & 0.52 & 0.24 & & 14244 & 10.17 & & $\infty$ & 76922 & & $\infty$ & $\infty$ & & $\infty$ & $\infty$\\

\midrule

\multirow{6}{*}{$\cS_{cd}$} & SHD & 1.00 & 1.00 & & 0.98 & 0.98 & & 0.70 & 0.76 & & 0.55 & 0.44 & & 0.53 & 0.46 & & 0.53 & 0.44 \\
& TP & 0.01 & 0.00 & & 0.04 & 0.02 & & 0.78 & 0.40 & & 1.00 & 0.98 & & 1.00 & 1.00 & & 1.00 & 1.00\\
& FP & 0.00 & 0.00 & & 0.00 & 0.00 & & 0.00 & 0.00 & & 0.00 & 0.00 & & 0.00 & 0.00 & & 0.00 & 0.00 \\
& FN & 0.99 & 1.00 & & 0.96 & 0.98 & & 0.22 & 0.60 & & 0.00 & 0.02 & & 0.00 & 0.00 & & 0.00 & 0.00 \\
& TPR & 0.01 & 0.00 & & 0.04 & 0.02 & & 0.78 & 0.40 & & 1.00 & 0.98 & & 1.00 & 1.00 & & 1.00 & 1.00\\
& FR & 0.24 & 0.07 & & 0.33 & 0.12 & & 6.13 & 1.38 & & 8763 & 63.69 & & $\infty$ & 175437 & & $\infty$ & $\infty$\\

\midrule

\multirow{6}{*}{$\cS_{dc}$} & SHD & 1.00 & 1.00 & & 0.96 & 0.99 & & 0.53 & 0.72 & & 0.50 & 0.56 & & 0.54 & 0.46 & & 0.53 & 0.56 \\
& TP & 0.01 & 0.01 & & 0.06 & 0.01 & & 0.91 & 0.48 & & 0.98 & 0.98 & & 1.00 & 1.00 & & 1.00 & 1.00\\
& FP & 0.00 & 0.00 & & 0.00 & 0.00 & & 0.00 & 0.00 & & 0.00 & 0.00 & & 0.00 & 0.00 & & 0.00 & 0.00 \\
& FN & 0.99 & 0.99 & & 0.94 & 0.99 & & 0.09 & 0.52 & & 0.02 & 0.02 & & 0.00 & 0.00 & & 0.00 & 0.00 \\
& TPR & 0.01 & 0.01 & & 0.06 & 0.01 & & 0.91 & 0.48 & & 0.98 & 0.98 & & 1.00 & 1.00 & & 1.00 & 1.00\\
& FR & 0.30 & 0.09 & & 0.37 & 0.12 & & 16.87 & 1.75 & & 12120 & 103.84 & & $\infty$ & 175437 & & $\infty$ & $\infty$\\

\midrule

& & \multicolumn{2}{c}{$\beta = 0.1$} & \phantom{abc}& \multicolumn{2}{c}{$\beta = 0.25$} &
\phantom{abc} & \multicolumn{2}{c}{$\beta = 0.4$} &
\phantom{abc} & \multicolumn{2}{c}{$\beta = 0.6$} &
\phantom{abc} & \multicolumn{2}{c}{$\beta = 0.75$} &
\phantom{abc} & \multicolumn{2}{c}{$\beta = 0.9$}\\
\cmidrule{3-4} \cmidrule{6-7} \cmidrule{9-10} \cmidrule{12-13} \cmidrule{15-16} \cmidrule{18-19}
\multirow{7}{*}{$\cS_{dd}$}  & Metric & RAG & DISC-2 && RAG & DISC-2 && RAG & DISC-2 && RAG & DISC-2 && RAG & DISC-2 && RAG & DISC-2 \\
\cmidrule{2-19}
& SHD & 0.97 & 1.00 & & 0.73 & 0.95 & & 0.51 & 0.77 & & 0.44 & 0.30 & & 0.41 & 0.30 & & 0.48 & 0.38 \\
& TP & 0.06 & 0.01 & & 0.48 & 0.25 & & 0.97 & 0.73 & & 1.00 & 0.99 & & 1.00 & 0.99 & & 1.00 & 1.00\\
& FP & 0.00 & 0.00 & & 0.00 & 0.00 & & 0.00 & 0.00 & & 0.00 & 0.00 & & 0.00 & 0.00 & & 0.00 & 0.00 \\
& FN & 0.94 & 0.99 & & 0.52 & 0.75 & & 0.03 & 0.27 & & 0.00 & 0.01 & & 0.00 & 0.00 & & 0.00 & 0.00 \\
& TPR & 0.06 & 0.01 & & 0.48 & 0.25 & & 0.97 & 0.73 & & 1.00 & 0.99 & & 1.00 & 1.00 & & 1.00 & 1.00\\
& FR & 0.92 & 0.20 & & 2.65 & 0.82 & & 27.38 & 4.54 & & 14597 & 197.41 & & 2499999 & 108694 & & $\infty$ & $\infty$\\
\bottomrule
\end{tabular}
}
\caption{The average metrics across 100 data replications for 2 nodes BN under $\cS_{cc}$, $\cS_{cd}$, $\cS_{dc}$ and $\cS_{dd}$. The data was generated according to the setting in Section \ref{subsubsec:sim.2nodes}.} \label{tab:2nodes.results} 
\end{table*}

The remaining scenarios presented in Table~\ref{tab:2nodes.results} exhibit a similar pattern as seen under $\cS_{cc}$. For hybrid Bayesian networks, $\cS_{cd}$ and $\cS_{dc}$, it is not surprising that RAG outperforms DISC since RAG exhibits no information loss. Under the $\cS_{dd}$ scenario containing data with two discrete variables, it is anti-intuitive at the first glance that RAG is better than DISC. The explanation is that $P(\cG_1 \mid \bX) \rightarrow \infty$ using BGe score as $\beta \rightarrow 1$, while BDe score has an upper limit. This is consistent with our theoretical results shown in Theorem \ref{thm:ratio.sdd}. Theorem \ref{thm:ratio.sdd} indicates that $\tilde{r}_{10}$ has an upper bound, while $r_{10}$ does not.

In summary, all the four scenarios imply that in a hybrid Bayesian network, RAG is a better choice compared to the discretization strategy.

\subsubsection{Results of BN with 4 Nodes}
\begin{table}[tb]
\centering
\resizebox{0.7\textwidth}{!}{%
\begin{tabular}{rrrrrrrrr}\toprule 

& \multicolumn{1}{r}{\textbf{Strategy}} & 
\multicolumn{1}{r}{\textbf{Score}} & 
\multicolumn{1}{r}{\textbf{SHD}} & 
\multicolumn{1}{r}{\textbf{TP}} &
\multicolumn{1}{r}{\textbf{FP}} &
\multicolumn{1}{r}{\textbf{FN}} &
\multicolumn{1}{r}{\textbf{TPR}} &
\multicolumn{1}{r}{\textbf{FR}} \\[0.5ex]\midrule

\parbox[t]{5mm}{\multirow{3}{*}{$\cS_{cc}$}} &  RAG &  BGe &  \textbf{0.03} &  \textbf{4.00} &  \textbf{0.02} &  \textbf{0.00} &  \textbf{1.00} &  \textbf{1.33}\\


    
& DISC-2 & BDe & 5.28 & 2.10 & 1.70 & 1.90 & 0.53 & 0.02\\

& DISC-4 & BDe & 4.72 & 2.00 & 1.01 & 2.00 & 0.50 & 0.00\\

\midrule

\parbox[t]{5mm}{\multirow{3}{*}{$\cS_{cd}$}} &  RAG &  BGe &  \textbf{0.57} &  \textbf{3.92} &  \textbf{0.26} &  \textbf{0.08} &  \textbf{0.99} &  \textbf{0.58}\\


    
& DISC-2 & BDe & 5.26 & 1.19 & 1.90 & 2.81 & 0.30 & 0.00\\

& DISC-4 & BDe & 5.47 & 1.27 & 1.97 & 2.73 & 0.32 & 0.00\\

\midrule

\parbox[t]{5mm}{\multirow{3}{*}{$\cS_{dc}$}} &  RAG &  BGe &  \textbf{0.21} &  \textbf{3.96} &  0.09 &  \textbf{0.05} &  \textbf{0.99} &  0.88\\


    
& DISC-2 & BDe & 0.80 & 3.54 & \textbf{0.05} & 0.46 & 0.89 & 0.74\\

& DISC-4 & BDe & 0.85 & 3.59 & 0.17 & 0.41 & 0.90 & \textbf{1.51}\\

\midrule

\parbox[t]{5mm}{\multirow{2}{*}{$\cS_{dd}$}} &  RAG &  BGe &  \textbf{0.93} & \textbf{3.73} &  0.05 & \textbf{0.27} &  \textbf{0.93} &  \textbf{0.13}\\

& - & BDe & 2.17 & 3.05 & \textbf{0.01} & 0.95 & 0.76 & 0.00\\

\bottomrule
\end{tabular}
}
\caption{The average metrics across 100 data replications for 4 nodes BN under $\cS_{cc}$, $\cS_{cd}$, $\cS_{dc}$ and $\cS_{dd}$.} \label{tab:4nodes.results} 
\end{table}
The results for 4-node BNs under the four data generation scenarios are summarized in Table \ref{tab:4nodes.results}. 
This table again confirms that the RAG is always the better choice when dealing with hybrid Bayesian networks, comparing to discretization strategy. The performance of DISC-2 against DISC-4 indicates that structure learning of hybrid Bayesian networks is not likely to be sensitive to the levels in discrete variables. It is noticeable that in some scenarios, such as $\cS_{dc}$ and $\cS_{dd}$, RAG is more likely to pick more false positives than DISC. One explanation to this observation is that RAG picks up slightly more edges than DISC. When there is a trade-off between controlling FP and FN, it is usually more preferable in practice to have higher FP than FN, since we can be more tolerant with false causal links than missing of any true causal link.
\section{Real Example Analysis}\label{sec:realdata}
Childhood obesity has been a serious public health issue worldwide \citep{chooi2019epidemiology}, affecting many people across their lifespans. To better understand the relationship of the variables surrounding childhood obesity for early interventions, we show an application of structure learning of Bayesian network to study this complex system for discovery of the multiple possible pathways that lead to childhood obesity.

The data used for analysis in this section is the `Growing Up in Australia: The \emph{Longitudinal Study of Australian Children}' (LSAC) \citep{mohal2021}. The LSAC is Australia's nationally representative children’s longitudinal study, focusing on social, economic, physical and cultural impacts on health, learning, social and cognitive development. The study tracks children’s development and life course trajectories in today’s economic, social and political environment. The aims of the data are to identify policy opportunities for improving support for children and their families as well as identify opportunities for early intervention.
\begin{table}[tb]
    \centering
    \resizebox{0.8\textwidth}{!}{%
    
    \begin{tabular}{llp{0.7\linewidth}} \toprule
        Name & Type  & Description  \\\midrule
         BMI & Continuous    & The z-score of BMI according to CDC growth charts. \\ 
         BM1 & Continuous    &BMI of Parent 1 (usually mother). \\ 
         BM2 & Continuous    &BMI of Parent 2 (usually father). \\ 
         INC & Continuous & Usual weekly income for household. \\ 
         SE  & Continuous   &The z-score for socioeconomic position among all families. \\
         BWZ & Continuous   & Birth weight Z-score. \\ 
         AC  & Discrete   &Child's choice (1 to 3) to spend free time. \\ 
         DP1 & Discrete   &Parent 1 depression K6 score. \\ 
         FH  & Discrete   & Household financial hardship score (0-6). \\ 
         FS  & Discrete  &Parent 1 financial stress (1-6). \\ 
         ME1 & Discrete   &Maternal high school education. \\ 
         ME2 & Discrete   & Maternal post-secondary education \\ 
         FE1 & Discrete   &Father's high school education. \\ 
         OD  & Discrete   & The quality of outdoor environment.\\ 
         RP1 & Discrete  & The scale of parent 1 feeling rushed.\\ 
         SL  & Discrete  & The study child sleep quality.\\ 
         SX  & Discrete  & Gender of child. \\ 
         
         GW & Discrete  & Gestation weeks. \\ 
         FV & Discrete  & Serves of fruit and vegetables per day. \\
         HF & Discrete  & Serves of high fat food (inc. whole milk) per day. \\ 
         HSD & Discrete & Serves of high sugar drinks per day. \\ 
         \bottomrule
    \end{tabular}
    }
    \caption{Description of variables that were used to construct the Bayesian network.}
    \label{Tab:variables}
\end{table}

Data from Wave 2 in Baby cohort is chosen to demonstrate the performance of different approaches. Data pre-processing was conducted to remove missing values and implausible entries, such as BMI values which are larger than 60, resulting in a data set of 2344 samples. Table \ref{Tab:variables} lists the variables and their descriptions used to construct the Bayesian networks, which contain both continuous and discrete random variables. Prior to implementing partition MCMC, it is necessary to forbid some links appearing in this application. For example, the link `Socio-Economic $\rightarrow$ SEX' does not make sense, because socio-economic can not drive the gender of a child. Table \ref{tab:blacklist} lists all forbidden links. 

\begin{table}[tb]
\centering
\resizebox{0.85\textwidth}{!}{%
\begin{tabular}{rcccccccccccccccccccc}
  \toprule
 & SX & BMI & SE & AC & INC & FS & FH & ME1 & FE1 & BM1 & BM2 & RP1 & DP1 & FV & HF & HSD & SL & OD & GW & BWZ \\ 
  \midrule
SX & & & & & & & & \xmark & \xmark & & & & & & & & & & \xmark & \xmark \\ 
BMI & \xmark & & \xmark & \xmark & & & & \xmark & \xmark & \xmark & \xmark & & & & & & & & \xmark & \xmark \\ 
SE & \xmark &  &  &  &  &  &  & \xmark & \xmark &  &  &  &  &  &  &  &  &  &  &  \\ 
AC & \xmark &  & \xmark &  &  &  &  & \xmark & \xmark &  &  &  &  &  &  &  &  &  & \xmark & \xmark \\ 
INC & \xmark &  &  &  &  &  &  & \xmark & \xmark &  &  &  &  &  &  &  &  &  &  &  \\ 
FS & \xmark &  &  &  &  &  &  & \xmark & \xmark &  &  &  &  &  &  &  &  &  & \xmark & \xmark \\ 
FH & \xmark &  &  &  &  &  &  & \xmark & \xmark &  &  &  &  &  &  &  &  &  & \xmark & \xmark \\ 
ME1 & \xmark &  &  &  &  &  &  &  &  &  &  &  &  &  &  &  &  &  &  &  \\ 
FE1 & \xmark &  &  &  &  &  &  &  &  &  &  &  &  &  &  &  &  &  &  &  \\ 
BM1 & \xmark &  &  &  &  &  &  &  &  &  &  &  &  &  &  &  &  &  &  &  \\ 
BM2 & \xmark &  &  &  &  &  &  &  &  &  &  &  &  &  &  &  &  &  &  &  \\ 
RP1 & \xmark &  &  &  &  &  &  & \xmark & \xmark &  &  &  &  &  &  &  &  &  & \xmark & \xmark \\ 
DP1 & \xmark &  &  &  &  &  &  & \xmark & \xmark &  &  &  &  &  &  &  &  &  & \xmark & \xmark \\ 
FV &  &  &  &  &  &  &  &  &  &  &  &  &  &  &  &  &  &  &  &  \\ 
HF &  &  &  &  &  &  &  &  &  &  &  &  &  &  &  &  &  &  &  &  \\ 
HSD &  &  &  &  &  &  &  &  &  &  &  &  &  &  &  &  &  &  &  &  \\ 
SL & \xmark &  & \xmark &  &  &  &  & \xmark & \xmark &  &  &  &  &  &  &  &  &  & \xmark & \xmark \\ 
OD & \xmark &  & \xmark &  &  &  &  & \xmark & \xmark &  &  &  &  &  &  &  &  &  & \xmark &  \\ 
GW & \xmark &  & \xmark &  &  &  &  & \xmark & \xmark & \xmark & \xmark &  &  &  &  &  &  &  &  &  \\ 
BWZ & \xmark &  & \xmark &  &  &  &  & \xmark & \xmark & \xmark & \xmark &  &  &  &  &  &  &  & \xmark &  \\ 
\bottomrule
\end{tabular}
} 
\caption{Forbidden links adopted in implementation of partition MCMC. The link always start from the row names to column names. The symbol \xmark \ indicates the link is forbidden.}
\label{tab:blacklist}
\end{table}

We implemented RAG and DISC respectively to the data using partition MCMC. The resulting most probable DAG structures produced under RAG and discretization strategy are shown in Figure \ref{fig:LSAC-RAG} and \ref{fig:LSAC-DISC} respectively. It is clear that the DAG produced by RAG strategy is much denser than that by DISC, indicating that DISC has missed many links which might have existed in reality. These missing edges are very likely to result in misunderstanding about the whole complex system. For example, according to the Figure \ref{fig:LSAC-RAG}, BMI is associated with BM1, BM2, and BWZ, while in Figure \ref{fig:LSAC-DISC} only BM1 (mother's BMI) is associated with BMI. Thus, misleading conclusions are likely to be concluded based on Figure \ref{fig:LSAC-DISC}. This real data analysis again confirms our findings in the simulation studies. That is, when dealing with hybrid Bayesian networks, RAG strategy is always preferred.

\begin{figure}[tb]
    \centering
    \begin{minipage}{0.45\textwidth}
        \centering
        \includegraphics[width=\textwidth]{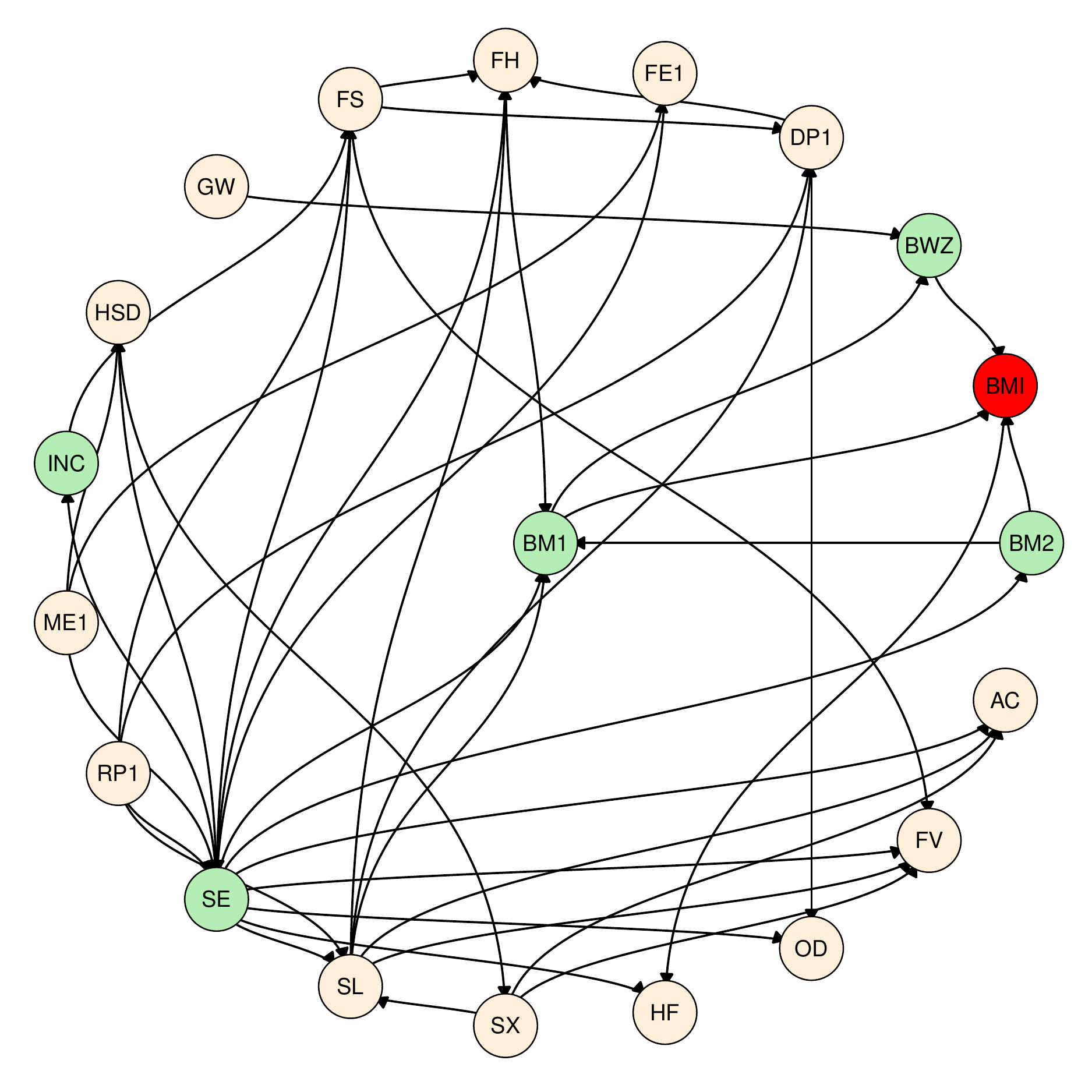} 
        \subcaption{RAG}\label{fig:LSAC-RAG}
    \end{minipage}\hfill
    \begin{minipage}{0.45\textwidth}
        \centering
        \includegraphics[width=\textwidth]{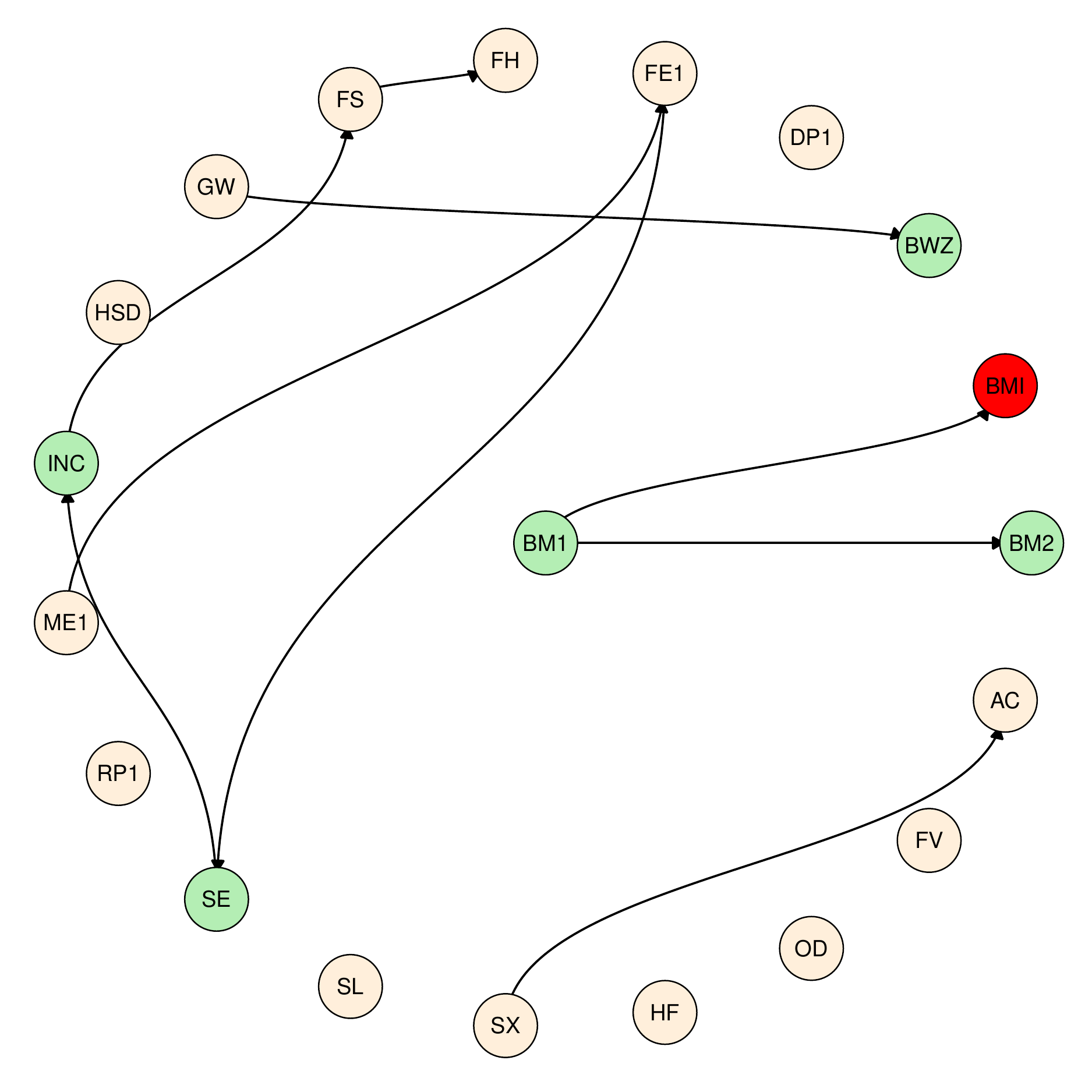} 
        \subcaption{DISC}\label{fig:LSAC-DISC}
    \end{minipage}
    \caption{The most likely DAGs under the two strategies. Continuous variables are represented by green nodes, while light yellow nodes depict discrete variables. The red node is the children's BMI.}
\end{figure}

To be complementary to the simulation studies, we also conduct further simulations based on this real data example. Two sets of simulations are conducted based on the DAG structures in Figure \ref{fig:LSAC-RAG} and \ref{fig:LSAC-DISC} respectively. In particular, assuming the DAG structure in Figure \ref{fig:LSAC-RAG} is the ground truth, we simulate 50 replicates of the data set, denoted by $\{\bX_{RAG}\}$. For each data set in $\{\bX_{RAG}\}$, we run partition MCMC with RAG and DISC respectively and compare their performance. Similarly, let $\{\bX_{DISC}\}$ denote the simulated 50 replicates of the data set based on the DAG in Figure \ref{fig:LSAC-DISC}. We run partition MCMC with RAG and DISC respectively and compare their performance. Note that each of the generated data set contains both discrete and continuous variables.

\begin{table}[tb]
\centering
\resizebox{0.7\textwidth}{!}{%
\begin{tabular}{ccccccccc}\toprule 

\multicolumn{1}{c}{\textbf{Data}}& \multicolumn{1}{c}{\textbf{Strategy}} & 
\multicolumn{1}{c}{\textbf{Score}} & 
\multicolumn{1}{c}{\textbf{SHD}} & 
\multicolumn{1}{c}{\textbf{TP}} &
\multicolumn{1}{c}{\textbf{FP}} &
\multicolumn{1}{c}{\textbf{FN}} &
\multicolumn{1}{c}{\textbf{TPR}} \\[0.5ex]\midrule

{\multirow{2}{*}{$\{\bX_{RAG}\}$}} &  RAG &  BGe & 3.84  & 35.98  & 0.08  & 1.02  & 0.97 \\
    
& DISC & BDe & 29.80 & 8.40 & 0.00 & 28.60 & 0.23 \\\midrule

{\multirow{2}{*}{$\{\bX_{DISC}\}$}} &  RAG &  BGe & 1.18  & 9.00  & 0.46  & 0.00  & 1.00  \\
    
& DISC & BDe & 0.54 & 9.00 & 0.00 & 0.00 & 1.00 \\

\bottomrule
\end{tabular}
}
\caption{The metrics of RAG and DISC under $\bX_{RAG}$ and $\bX_{DISC}$ respectively.} \label{tab:LSAC.simulation} 
\end{table}
We apply both RAG and DISC strategies to the generated data $\{\bX_{RAG}\}$ and $\{\bX_{DISC}\}$. Table \ref{tab:LSAC.simulation} summaries the results of the above simulations. The averaged numbers of false negatives under $\{\bX_{RAG}\}$ scenario indicates that RAG strategy significantly outperforms DISC, while DISC strategy has missed many links. Considering the data under $\{\bX_{DISC}\}$, RAG's performance is comparable with that of DISC, with marginally more false positives. As it is known, we can not simultaneously control FP and FN. Although the performance between the two strategies are similar in terms of controlling FP, RAG is significantly better in controlling FN. In addition, it can be argued that FP is more tolerable than FN in the structure learning problem, since all the true links should be detected. Therefore, RAG should be the preferred strategy over DISC in practice. 

\section{Discussion and Conclusion} \label{sec:conclusion}

We reviewed the literature on structure learning of hybrid BNs, in which either the relationship between discrete and continuous variables are directly modeled to develop a score function for searching or discretization strategies are employed to convert a hybrid Bayesian network into a well-studied discrete Bayesian network. 

Our proposed RAG strategy considers all discrete variables in a hybrid BN as continuous variables, thus converting hybrid networks into continuous Bayesian networks that can be handled by many existing techniques. Although RAG mis-specifies the distributions of the discrete variables, this has minor impact on structure learning, and the benefits of RAG are significant. Compared to the methods discussed in Section~\ref{sec:intro}, RAG saves significant computational resources. RAG also has no information loss when compared to discretization techniques and we do not need to worry about implementing discretization. We have shown theoretically and via simulation the advantages of RAG compared to discretization strategy under a score-based structure learning approach. Through the theoretical and empirical analysis, we conclude that RAG is the preferable choice when dealing with hybrid Bayesian networks, comparing with discretization techniques.

Throughout the paper, we assume the discrete random variables are ordinal random variables. In practice, the discrete variables could be nominal variables. For example, a variable representing colors can take values from a set \{red, blue, yellow\}. For such circumstances, there is no standard way to convert nominal variables to a numeric variable. One solution is to convert the nominal variable to dummy variables, which will be treated as continuous. A $q$-level nominal variable will be replaced by $q-1$ dummy variables. Although this comes at the cost of increasing the number of nodes in the resulting Bayesian networks and thus requiring more computational resources, this solution makes RAG feasible to handle nominal variables.

When using the DISC strategy, a continuous random variable is always discretized to a binary random variable throughout this paper. The motivation is largely to keep our theoretical results succinct. Otherwise, it becomes too tedious to derive the theoretical results, as it requires to write out all the probabilities $P(X'_1 =i, X'_2 = j)$. And the number of probabilities grows quadratically with respect to the number of levels in the discretized random variables.

\bibliographystyle{ba}
\bibliography{main}

\newpage
\section*{Supplementary Materials for ``Structure Learning for Hybrid Bayesian Networks''}
\appendix

\input{Appendix}

\end{document}

%% file: Appendix.tex
\section{Proof of Lemma \ref{lemma:2nodes}} \label{proof:lemma.2nodes}

\begin{proof}

For any graph, it follows,
\begin{eqnarray} \label{eqn:decomposable}
P(\cG\mid \bX) \propto P(\bX \mid \cG) = \prod_{i=1}^n P(X_i \mid \textbf{Pa}_i^{\cG}) =\left(\prod_{i\in O_\cG} P(X_i) \right) \left(\prod_{i \notin O_\cG} P(X_i \mid \textbf{Pa}_i^{\cG}) \right),
\end{eqnarray}
where $O_\cG$ denotes the set of outpoints in $\cG$.
As $\#(\textbf{Pa}_i^{\cG}) >0$ for $i \notin O_\cG$,  without loss generality, we assume there are $m_i$ elements in $\textbf{Pa}_i^{\cG}$ and denoted by $\textbf{Pa}_i^{\cG} = \{Z_{i,1}, \cdots, Z_{i,m_i}\}$. In this case, we have
\begin{eqnarray} \label{eqn:seqprod}
    P(X_i \mid \textbf{Pa}_i^{\cG}) = P(X_i \mid Z_{i,1},\cdots, Z_{i,m_i}) = P(Y_{i,1} \mid Z_{i,1}),
\end{eqnarray}
where $Y_{i,1} = (X_i \mid Z_{i,2},\cdots,Z_{i,m_i})$ denotes the residual of $X_i$ eliminating the impacts from $Z_{i,2},\cdots,Z_{i,m_i}$. Recursively, we have 
\begin{eqnarray}
Y_{i,t} = (Y_{i,t+1} \mid Z_{i,t+1}), \quad \quad t= 1,\cdots, m_i-1, 
\end{eqnarray}
where $Y_{i,t}= (X_i \mid Z_{i,t+1},\cdots, Z_{i,m_i})$. Thus, Equation \ref{eqn:seqprod} can be achieved by conducting regressions starting from $t=m_i - 1$ to $1$, each of which is a simple linear regression.

Considering, the terms in $\prod_{i \notin O_\cG} P(X_i \mid \textbf{Pa}_i^{\cG}) $ are mutually independent, $P(\cG\mid \bX)$ can be derived by conducting a series of simple linear regressions. Actually, the number of simple linear regressions required in computation of $P(\cG\mid \bX)$ is equal to the number of edges in $\cG$.

\end{proof}

\section{Proof of Lemma~\ref{lemma:BGe.posterior.ratio}}
\begin{proof} \label{proof:bge.posteriorratio}


We have set a joint prior distribution on $\bW$ and $\bmu$, which are shown in Equation \eqref{eqn:prior.W} and \eqref{eqn:prior.mu}.
The marginal distribution of the data $\bX$ is computed as,
\begin{eqnarray*}
P(\bX)&=& \int P(\bX \mid \bmu, \bW) P(\bmu,\bW) d\bmu d \bW = \left(\frac{\alpha_\mu}{N+\alpha_\mu}\right)^{\frac{n}{2}} \frac{\Gamma_n(\frac{N+\alpha_w}{2})}{\pi^{\frac{nN}{2}} \Gamma_n(\frac{\alpha_w}{2})} \frac{\mid \bT\mid ^{\frac{\alpha_w}{2}}}{\mid \bR\mid ^{\frac{N+\alpha_w}{2}}},
\end{eqnarray*}
where $\bR=\bT+\bS_N+\frac{N\alpha_\mu}{N+\alpha_\mu} (\boldsymbol{\bar{x}} - \bnu)(\boldsymbol{\bar{x}} - \bnu)^T$, $\boldsymbol{\bar{x}} = \frac{1}{N} \sum_{i=1}^N \boldsymbol{x}_{i}$, $\bS_N= \sum_{i=1}^N (\bx_i - \boldsymbol{\bar{x}})(\bx_i - \boldsymbol{\bar{x}})^T = N \Sigma$, $\Sigma$ is the covariance matrix of $(X_1, X_2)$ and $\bx_i = (x_{1i},x_{2i})^T$.

Let $\bX^Y$ denote the subset of $\bX$. For example, if $Y=\{X_1\}$, then $\bX^Y$ denotes the column $X_1$ in data $\bX$.
Similarly, the marginal distribution of a subset of the data $\bX$ is computed as,
\begin{eqnarray*}
P(\bX^Y)= \left(\frac{\alpha_\mu}{N+\alpha_\mu}\right)^{\frac{l}{2}} \frac{\Gamma_l(\frac{N+\alpha_w-n+l}{2})}{\pi^{\frac{lN}{2}} \Gamma_l(\frac{\alpha_w-n+l}{2})} \frac{\mid \bT_{YY}\mid ^{\frac{\alpha_w-n+l}{2}}}{\mid \bR_{YY}\mid ^{\frac{N+\alpha_w-n+l}{2}}},
\end{eqnarray*}
where $l$ is the number of variables in $Y$, $\bT_{YY}$ and $\bR_{YY}$ denote the sub-matrices of $\bT$ and $\bR$ respectively, corresponding to the subset $Y$.

The logarithm of posterior ratio between $P(\cG_1 \mid \bX )$ and $P(\cG_0 \mid \bX )$ is
\begin{eqnarray}
&&\log\left(\frac{P(\cG_1 \mid \bX )}{P(\cG_0 \mid \bX )}\right) = \log\left(\frac{P( \bX \mid \cG_1 ) P(\cG_1)}{P( \bX \mid \cG_0 ) P(\cG_0)} \right)= \log\left(\frac{P( \bX \mid \cG_1 ) }{P( \bX \mid \cG_0 ) }\right) \notag\\
&&= \log\left(\frac{P( \bX ) }{P( \bX^{X_1}) P( \bX^{X_2})}\right) = \log \left( \frac{\Gamma(\frac{N+\alpha_w}{2}) \Gamma(\frac{\alpha_w-1}{2})}{ \Gamma(\frac{\alpha_w}{2}) \Gamma(\frac{N+\alpha_w-1}{2})} 
\frac{\mid \bT\mid ^{\frac{\alpha_w}{2}}}{\mid \bR\mid ^{\frac{N+\alpha_w}{2}}} 
\frac{( \bR_{11} \bR_{22} ) ^{\frac{N+\alpha_w-1}{2}}} {( \bT_{11} \bT_{22}) ^{\frac{\alpha_w-1}{2}}} \right) \notag \\
&&=C -\frac{\alpha_w +N}{2}\log(|\bR|) +\frac{N-1+\alpha_w }{2} \log(\bR_{11} \bR_{22}), \notag
\end{eqnarray}
where $C=\log\left(\frac{\Gamma(\frac{\alpha_w+N}{2}) \Gamma(\frac{\alpha_w-1}{2}) }{\Gamma(\frac{\alpha_w+N-1}{2}) \Gamma(\frac{\alpha_w}{2})}\right) +\log(t)$, $\bR=\bT+\bS_N+\frac{N\alpha_\mu}{N+\alpha_\mu} (\boldsymbol{\bar{x}} - \bnu)(\boldsymbol{\bar{x}} - \bnu)^T$, and $\bR_{ij}$ denotes the element at the  $i^{th}$ row and $j^{th}$ column. 

For a large $N$, the matrix $\bR$ is dominated by $\bS_N$ and $C$ converges to a constant. Thus, we have 
\begin{eqnarray}
&&\lim_{N\rightarrow \infty}\frac{1}{N} \log\left(\frac{P(\cG_1 \mid \bX )}{P(\cG_0 \mid \bX )}\right) \notag \\
&=& \lim_{N\rightarrow \infty}\frac{1}{N} \left(C -\frac{\alpha_w +N}{2}\log(|\bR|) +\frac{N-1+\alpha_w }{2} \log(\bR_{11} \bR_{22}) \right) \notag \\
&=&  \lim_{N\rightarrow \infty} \left( -\frac{1}{2}\log(|\bR|) +\frac{1 }{2} \log(\bR_{11} \bR_{22}) \right) \notag \\
&=& \frac{1}{2} \log\left( \frac{\Sigma_{11} \Sigma_{22}}{\Sigma_{11} \Sigma_{22} - \Sigma_{12}^2} \right),\label{eqn:log.ratio}
\end{eqnarray}
where $\Sigma_{ij}$ is the element of covariance matrix $\Sigma$. Equation \eqref{eqn:log.ratio} indicates that as sample size $N$ goes to infinity the average posterior ratio is determined by covariance matrix $\Sigma$.

Substituting Equation \eqref{eqn:log.ratio} into Equation \eqref{eqn:evaluation.criterion}, we have
\begin{eqnarray*}
\lim_{N \rightarrow \infty} r_{10} &=& \lim_{N \rightarrow \infty} \frac{1}{N} \int_{\bx \in \Omega_{\bX}} \log\left(\frac{P(\cG_1 \mid \bX )}{P(\cG_0 \mid \bX )}\right) P(\bX) d\bx\\
&=&\frac{1}{2} \log\left( \frac{\Sigma_{11} \Sigma_{22}}{\Sigma_{11} \Sigma_{22} - \Sigma_{12}^2} \right) \int_{\bx \in \Omega_{\bX}} P(\bX) d\bx\\
&=&\frac{1}{2} \log\left( \frac{\Sigma_{11} \Sigma_{22}}{\Sigma_{11} \Sigma_{22} - \Sigma_{12}^2} \right).
\end{eqnarray*}

The proof completes.
\end{proof}

\section{Proof of Lemma~\ref{lemma:BDe.posterior.ratio} }
\begin{proof} \label{proof:bde.posteriorratio}
We have set a joint prior distribution on $\theta_{12}, \theta_{1 \bar{2}}, \theta_{\bar{1}2}$, which is shown in Equation \eqref{eqn:prior.BDe}. To compute the marginal distribution of $\bX$, we need to derive the prior distributions on the parameters attached to $\cG_1$ and $\cG_0$ respectively.

For graph $\cG_1$, i.e., graph $X_1 \rightarrow X_2$, the parameters will be $\theta_1, \theta_{2\mid 1},\theta_{2\mid \bar{1}}$, where $\theta_1 = P(X_1 =1), \theta_{2\mid 1} = P(X_2 =1 \mid X_1 =1)$ and $\theta_{2\mid \bar{1}} = P(X_2 =1 \mid X_1 =0)$. Then the prior distribution of $\theta_1, \theta_{2\mid 1},\theta_{2\mid \bar{1}}$ will be 
\begin{eqnarray*}
P(\theta_1, \theta_{2\mid 1},\theta_{2\mid \bar{1}} \mid \cG_1, \alpha_1,\cdots,\alpha_4) = J_{\cG_1} P(\theta_{12},\theta_{1\bar{2}},\theta_{\bar{1}2} \mid \alpha_1,\cdots,\alpha_4),
\end{eqnarray*}
where $J_{\cG_1}$ is the Jacobian matrix between the transformation from $\theta_{12},\theta_{1\bar{2}},\theta_{\bar{1}2}$ to $\theta_1, \theta_{2\mid 1},\theta_{2\mid \bar{1}}$. By definition, $J_{\cG_1} = \theta_1 (1-\theta_1).$
Then,
\begin{eqnarray}
&&P(\theta_1, \theta_{2\mid 1},\theta_{2\mid \bar{1}} \mid \cG_1, \alpha_1,\cdots,\alpha_4) \notag \\
&&= J_{\cG_1} P(\theta_{12},\theta_{1\bar{2}},\theta_{\bar{1}2} \mid \alpha_1,\cdots,\alpha_4) \notag \\
&&=\theta_1 (1-\theta_1) \frac{1}{\boldsymbol{B}(\balpha)} (\theta_{1}\theta_{2\mid 1})^{\alpha_1 - 1} (\theta_{1} (1-\theta_{2\mid 1}))^{\alpha_2 -1 } ((1-\theta_{1})\theta_{2\mid \bar{1}})^{\alpha_3 -1} ((1-\theta_{1})(1-\theta_{2\mid \bar{1}}))^{\alpha_4 -1} \notag\\
&&= \frac{1}{\boldsymbol{B}(\balpha)} \theta_1^{\alpha_1+ \alpha_2 -1} (1-\theta_1)^{\alpha_3+\alpha_4-1} \theta_{2\mid {1}}^{\alpha_1 -1} (1-\theta_{2\mid {1}})^{\alpha_2 -1} \theta_{2\mid \bar{1}}^{\alpha_3 -1} (1-\theta_{2\mid \bar{1}})^{\alpha_4 -1}, \label{eqn:prior.g12}
\end{eqnarray}
where $\balpha = (\alpha_1, \alpha_2, \alpha_3, \alpha_4)$, $\boldsymbol{B}(\balpha)$ is a multivariate beta function, i.e., $\boldsymbol{B}(\balpha) = \frac{\prod_{i=1}^4 \Gamma(\alpha_i)}{\Gamma(\sum_{i=1}^4 \alpha_i)}$. For a general vector $\boldsymbol{v} = (v_1, \cdots, v_K)$, $\boldsymbol{B}(\boldsymbol{v}) = \frac{\prod_{i=1}^K \Gamma({v}_i)}{\Gamma(\sum_{i=1}^K {v}_i)}$. When $K=2$, the multivariate beta function is reduced to beta function.

Similarly, we can derive the prior distribution for graph $\cG_{2}: X_2 \rightarrow X_1$, which is an equivalence class of $\cG_1$, as follows
\begin{eqnarray}
&&P(\theta_2, \theta_{1\mid 2},\theta_{1\mid \bar{2}} \mid \cG_{2}, \alpha_1,\cdots,\alpha_4) \notag \\
&&= \frac{1}{\boldsymbol{B}(\balpha)} \theta_2^{\alpha_1+ \alpha_2 -1} (1-\theta_2)^{\alpha_3+\alpha_4-1} \theta_{1\mid {2}}^{\alpha_1 -1} (1-\theta_{1\mid {2}})^{\alpha_2 -1} \theta_{1\mid \bar{2}}^{\alpha_3 -1} (1-\theta_{1\mid \bar{2}})^{\alpha_4 -1}. 
\end{eqnarray}
In graph $\cG_0$, there are only two parameters $\theta_1$ and $\theta_2$. To derive the prior distribution $P(\theta_1, \theta_2 \mid \cG_0,  \alpha_1, \cdots, \alpha_4)$, we need the following two assumptions: global parameter independence and parameter modularity \citep{heckerman1995}. According to global parameter independence, we have
\begin{eqnarray}
P(\theta_1, \theta_2 \mid \cG_0, \alpha_1, \cdots, \alpha_4) = P(\theta_1 \mid \cG_0, \alpha_1, \cdots, \alpha_4) P(\theta_2 \mid \cG_0, \alpha_1, \cdots, \alpha_4). \label{eqn:parameter.independence}
\end{eqnarray}
According to parameter modularity, we have 
\begin{eqnarray}
P(\theta_1 \mid \cG_0, \alpha_1, \cdots, \alpha_4) &=& P(\theta_1 \mid \cG_1, \alpha_1, \cdots, \alpha_4), \label{eqn:modularity1} \\ 
P(\theta_2 \mid \cG_0, \alpha_1, \cdots, \alpha_4) &=& P(\theta_2 \mid \cG_{2}, \alpha_1, \cdots, \alpha_4). \label{eqn:modularity2}
\end{eqnarray}

Combining Equations \eqref{eqn:prior.g12}-\eqref{eqn:modularity2} together and computing the normalizing constant to make a proper density function,
we have
\begin{eqnarray}
&&P(\theta_1, \theta_{2} \mid \cG_0, \alpha_1,\cdots,\alpha_4) \nonumber\\
&&=\left[\frac{1}{\boldsymbol{B}(\alpha_1+\alpha_2,\alpha_3+\alpha_4)}\right]^2  \theta_1^{\alpha_1+ \alpha_2 -1} (1-\theta_1)^{\alpha_3+\alpha_4-1} \theta_2^{\alpha_1+ \alpha_2 -1} (1-\theta_2)^{\alpha_3+\alpha_4-1}. \label{eqn:prior.g0}
\end{eqnarray}

Once we have the prior distribution for $\cG_1$ (Equation \eqref{eqn:prior.g12}), integrating out all the parameters out, we have
\begin{eqnarray}
&&P(\bX \mid \cG_1) \notag \\
&=& \int P(\bX \mid \cG_1,\theta_1,\theta_{2\mid 1},\theta_{2\mid \bar{1}}) P(\theta_1, \theta_{2\mid 1},\theta_{2\mid \bar{1}} \mid \cG_1, \alpha_1,\cdots,\alpha_4) d\theta_1 d\theta_{2\mid 1} d \theta_{2\mid \bar{1}}\notag \\
&=&\frac{\boldsymbol{B}(\alpha_{12}+N_{1.}, \alpha_{34}+N_{0.}) \boldsymbol{B}(\alpha_1+N_{11}, \alpha_2 + N_{10}) \boldsymbol{B}(\alpha_3+ N_{01}, \alpha_4 + N_{00}) }{\boldsymbol{B}(\balpha)},\notag \\
&=& \frac{\Gamma(\alpha_1+N_{11}) \Gamma(\alpha_2 + N_{10}) \Gamma(\alpha_3+ N_{01})  \Gamma( \alpha_4 + N_{00}) }{\boldsymbol{B}(\balpha) \Gamma(\alpha_{12}+\alpha_{34}+N)}, \label{eqn:BDe.g1.likelihood}
\end{eqnarray}
where $N_{11}$ denotes the number of samples in $\bX$ equal to $(1,1)$, $N_{10}$ denotes the number of samples in $\bX$ equal to $(1,0)$, $N_{01}$ denotes the number of samples in $\bX$ equal to $(0,1)$, $N_{00}$ denotes the number of samples in $\bX$ equal to $(0,0)$, $N_{1.} = N_{11}+N_{10}$, $N_{0.} = N_{01}+N_{00}$, $\alpha_{12} = \alpha_1 + \alpha_2$ and $\alpha_{34} = \alpha_3 + \alpha_4$.

Similarly, for graph $\cG_0$, whose prior distribution on $\theta_1, \theta_2$ is shown by Equation \eqref{eqn:prior.g0}, integrating out all the parameters out, we have
\begin{eqnarray}
P(\bX \mid \cG_0) &=& \int P(\bX \mid \cG_0,\theta_1,\theta_{2}) P(\theta_1, \theta_{2} \mid \cG_0, \alpha_1,\cdots,\alpha_4) d\theta_1 d \theta_2\notag\\
&=& \left[\frac{1}{\boldsymbol{B}(\alpha_{12},\alpha_{34})}\right]^2  \int \theta_1^{\alpha_{12}+N_{1.} -1} (1-\theta_1)^{\alpha_{34} + N_{0.} -1} \theta_2^{\alpha_{12}+ N_{.1} -1} (1-\theta_2)^{\alpha_{34}+N_{.0}-1}d\theta_1 d \theta_2 \notag\\
&=& \left[\frac{1}{\boldsymbol{B}(\alpha_{12},\alpha_{34})}\right]^2 \boldsymbol{B}(\alpha_{12}+N_{1.},\alpha_{34} + N_{0.}) \boldsymbol{B}(\alpha_{12}+N_{.1},\alpha_{34} + N_{.0}) \notag\\
&=& \left[\frac{1}{\boldsymbol{B}(\alpha_{12},\alpha_{34})}\right]^2 \frac{\Gamma(\alpha_{12}+N_{1.}) \Gamma(\alpha_{34} + N_{0.})}{\Gamma(\alpha_{12}+\alpha_{34}+N)} \frac{\Gamma(\alpha_{12}+N_{.1}) \Gamma(\alpha_{34} + N_{.0})}{\Gamma(\alpha_{12}+\alpha_{34}+N)},\label{eqn:BDe.g0.likelihood}
\end{eqnarray}
where $N_{.1} = N_{01}+ N_{11}$, $N_{.0} = N_{00}+ N_{10}$.
Combining Equation \eqref{eqn:BDe.g1.likelihood} and \eqref{eqn:BDe.g0.likelihood}, we have 
\begin{eqnarray}
&&\lim_{N \rightarrow \infty}\frac{1}{N}\log\left(\frac{P(\cG_1 \mid \bX  )}{P(\cG_0 \mid \bX )}\right)=\lim_{N \rightarrow \infty}\frac{1}{N}\log\left(\frac{P(  \bX \mid \cG_1  )}{P( \bX \mid \cG_0 )}\right) \notag \\
&=&\lim_{N \rightarrow \infty} \frac{1}{N} \log \left( \frac{\left[\boldsymbol{B}(\alpha_{12},\alpha_{34})\right]^2 \Gamma(\alpha_{12}+\alpha_{34}+N) \Gamma(\alpha_1+N_{11}) \Gamma(\alpha_2 + N_{10}) \Gamma(\alpha_3+ N_{01})  \Gamma( \alpha_4 + N_{00}) }{\boldsymbol{B}(\balpha) \Gamma(\alpha_{12}+N_{1.}) \Gamma(\alpha_{34} + N_{0.}) \Gamma(\alpha_{12}+N_{.1}) \Gamma(\alpha_{34} + N_{.0})}\right)\notag \\
&=& \lim_{N \rightarrow \infty} \frac{1}{N} \log \left( \frac{ \Gamma(\alpha_{12}+\alpha_{34}+N) \Gamma(\alpha_1+N_{11}) \Gamma(\alpha_2 + N_{10}) \Gamma(\alpha_3+ N_{01})  \Gamma( \alpha_4 + N_{00}) }{ \Gamma(\alpha_{12}+N_{1.}) \Gamma(\alpha_{34} + N_{0.}) \Gamma(\alpha_{12}+N_{.1}) \Gamma(\alpha_{34} + N_{.0})}\right) \label{eqn:bde.ratio.p1}.
\end{eqnarray}
According to the Stirling's approximation \citep{Stirling2010}, for a large integer value $z$,
\begin{eqnarray} \label{eqn:bde.ratio.p2}
&&\lim_{N\rightarrow \infty} \frac{1}{N}{\log(\Gamma(z+1))} \notag \\ &=&\lim_{N\rightarrow \infty} \frac{1}{N}\left( z\log(z) - z +\frac{1}{2} \log(2\pi z) + \mathcal{O}\left(\frac{1}{z}\right) \right) \notag\\
&=& \lim_{N\rightarrow \infty} \frac{1}{N}\left( z\log(z) - z \right).
\end{eqnarray}

Combining Equation \eqref{eqn:bde.ratio.p1} and \eqref{eqn:bde.ratio.p2}, we have
\begin{eqnarray}
&&\lim_{N \rightarrow \infty}\frac{1}{N}\log\left(\frac{P(\cG_1 \mid \bX  )}{P(\cG_0 \mid \bX )}\right) \notag\\
&=&\lim_{N \rightarrow \infty}  \log(\alpha_{12}+\alpha_{34}+N-1)+p_{11} \log(\alpha_1+N_{11}-1) + p_{10} \log(\alpha_2+N_{10}-1) \notag\\
&& + p_{01} \log(\alpha_3+N_{01}-1) +p_{00} \log(\alpha_4+N_{00}-1) -p_{1.}\log(\alpha_{12}+N_{1.}-1) \notag \\
&& - p_{0.}\log(\alpha_{34}+N_{0.}-1) -p_{.1}\log(\alpha_{12}+N_{.1}-1) -p_{.0}\log(\alpha_{34}+N_{.0}-1) \notag \\
&=& \lim_{N \rightarrow \infty} p_{11} \log\left(\frac{\alpha_1 +N_{11}-1}{\alpha_{12} + N_{1.}-1}\right) + p_{10} \log\left(\frac{\alpha_2 +N_{10}-1}{\alpha_{12} + N_{1.}-1}\right) + p_{01} \log\left(\frac{\alpha_3 +N_{01}-1}{\alpha_{34} + N_{0.}-1}\right) \notag \\
&&+ p_{00} \log\left(\frac{\alpha_4 +N_{00}-1}{\alpha_{34} + N_{0.}-1}\right) +  p_{.1} \log\left(\frac{\alpha_{12}+\alpha_{34} +N-1}{\alpha_{12} + N_{.1}-1}\right)+ p_{.0} \log\left(\frac{\alpha_{12}+\alpha_{34} +N-1}{\alpha_{34} + N_{.0}-1}\right) \notag \\
&=& p_{11} \log\left(\frac{p_{11}}{p_{1.}}\right) + p_{10} \log\left(\frac{p_{10}}{  p_{1.}}\right) + p_{01} \log\left(\frac{ p_{01}}{ p_{0.}}\right) + p_{00} \log\left(\frac{ p_{00}}{p_{0.}}\right) +  p_{.1} \log\left(\frac{ 1}{p_{.1}}\right)+ p_{.0} \log\left(\frac{1}{p_{.0}}\right) \notag \\
&=& p_{11} \log\left(\frac{p_{11}}{p_{1.} p_{.1}}\right) + p_{10} \log\left(\frac{p_{10}}{  p_{1.} p_{.0} }\right) + p_{01} \log\left(\frac{ p_{01}}{ p_{0.} p_{.1}}\right) + p_{00} \log\left(\frac{ p_{00}}{p_{0.} p_{.0}}\right), \notag
\end{eqnarray}
where $p_{1.} = \frac{N_{11}+ N_{10}}{N} = p_{11}+p_{10}, p_{0.} = \frac{N_{01}+ N_{00}}{N} = p_{01}+p_{00}, p_{.1} = \frac{N_{11}+ N_{01}}{N} = p_{11}+p_{01}$ and $p_{.0} = \frac{N_{10}+ N_{00}}{N} = p_{10}+p_{00}$.
Therefore,
\begin{eqnarray*}
&&\lim_{N \rightarrow \infty}r_{10} = \lim_{N \rightarrow \infty} \sum_{\bx \in \Omega_{\bX}}  \frac{1}{N}\log\left(\frac{P(\cG_1 \mid \bX  )}{P(\cG_0 \mid \bX )}\right) P(\bX = \bx)\\
&=& p_{11} \log\left(\frac{p_{11}}{p_{1.} p_{.1}}\right) + p_{10} \log\left(\frac{p_{10}}{  p_{1.} p_{.0} }\right) + p_{01} \log\left(\frac{ p_{01}}{ p_{0.} p_{.1}}\right) + p_{00} \log\left(\frac{ p_{00}}{p_{0.} p_{.0}}\right).
\end{eqnarray*}
The proof completes.
\end{proof}

\section{Proof of Theorem~\ref{thm:ratio.scc}}
\begin{proof} \label{proof:ratio.scc}
As the data is generated under scenario $\cS_{cc}$, the covariance matrix can be computed as 
$\Sigma = \left[\begin{array}{cc}
    \sigma_{1}^2 &\beta \sigma_{1}^2  \\
    \beta \sigma_{1}^2  & \sigma_{2}^2 + \beta^2 \sigma_{1}^2
\end{array}\right]$.
Substituting the covariance matrix to Lemma \ref{lemma:BGe.posterior.ratio}, immediately we have
\begin{eqnarray*}
\lim_{N\rightarrow \infty} r_{10} =\frac{1}{2} \log\left( \frac{\Sigma_{11}\Sigma_{22}}{\Sigma_{11}\Sigma_{22} - \Sigma_{12}^2} \right) = \frac{1}{2} \log\left( \frac{\sigma_2^2 + \beta^2 \sigma_1^2}{\sigma_2^2} \right).
\end{eqnarray*}

For the transformed data $\bX'$, define  $\tilde{p}_{ij} \triangleq P(X_1^{'} = i, X_2^{'} =j), i \in \{0,1\},j \in \{0,1\}$.
Then we have
\begin{eqnarray}
&& \tilde{p}_{11} =P(X_2^{'}=1, X_1^{'}=1)  =P(X_2 \geq \tilde{\mu}_2 , X_1 \geq \mu_1) \nonumber\\
&&=\int_{\mu_1}^\infty \int_{\tilde{\mu}_2}^{\infty} \frac{1}{\sqrt{2\pi }\sigma_1  } \exp\left\{- \frac{(x_1-\mu_1)^2}{2 \sigma_1^2} \right\} \frac{1}{\sqrt{2\pi }\sigma_2  } \exp\left\{- \frac{(x_2-\mu_2-\beta x_1)^2}{2 \sigma_2^2} \right\} dx_1 dx_2 \nonumber\\
&& = \int_{\mu_1}^\infty  \frac{1}{\sqrt{2\pi }\sigma_1  } \exp\left\{- \frac{(x_1-\mu_1)^2}{2 \sigma_1^2} \right\}  \Phi \left(- \frac{(\beta \mu_1-\beta x_1)}{ \sigma_2} \right) dx_1 \notag \\
&&= \int_{0}^{\infty} \frac{1}{\sqrt{2\pi } } \exp\left\{- \frac{x^2}{2 } \right\} \Phi\left(\frac{\sigma_1\beta x}{\sigma_2}\right)dx \nonumber\\
&&= \int_{0}^{\infty} \frac{1}{\sqrt{2\pi } } \exp\left\{- \frac{x^2}{2 } \right\} \Phi\left(\frac{\sigma_1\beta x}{\sigma_2}\right)dx. \label{eqn:p11.scc}
\end{eqnarray}
Since the data $\bX'$ is obtained by discretizing the original continuous random variables into binary variables using the middle point rule, it is straightforward to conclude
\begin{eqnarray} 
P(X'_1 = 1) = P(X'_1 = 0) = P(X'_2 = 1) = P(X'_2 = 0) = \frac{1}{2}. \label{eqn:marginals}
\end{eqnarray}

Combining the Equation \eqref{eqn:marginals} with \eqref{eqn:p11.scc}, we have
\begin{eqnarray}
&&\tilde{p}_{10} = P(X_2^{'}=0, X_1^{'}=1) =P(X^{'}_1 =1)-P(X^{'}_2=1, X^{'}_1=1) =\frac{1}{2}- \tilde{p}_{11}. \label{eqn:p10.scc}\\ 
&& \tilde{p}_{01}=P(X^{'}_2=1, X^{'}_1=0) = P(X^{'}_2=1) - P(X^{'}_2=1, X^{'}_1 =1)=\frac{1}{2}-\tilde{p}_{11}. \label{eqn:p01.scc} \\
&&\tilde{p}_{00}= P(X^{'}_2 =0, X^{'}_1 =0)=P(X^{'}_2=0)- P(X^{'}_2=0, X^{'}_1=1) =\frac{1}{2}- \tilde{p}_{10} =\tilde{p}_{11}. \label{eqn:p00.scc} 
\end{eqnarray}
Given the new data $\bX^{'}$, the posterior ratio between $P(\cG_1 \mid \bX^{'})$ and $P(\cG_0 \mid \bX^{'})$ can be derived by combining Lemma~\ref{lemma:BDe.posterior.ratio} with Equations~\eqref{eqn:p11.scc}, \eqref{eqn:marginals}, \eqref{eqn:p10.scc}, \eqref{eqn:p01.scc} and \eqref{eqn:p00.scc}. 
\begin{eqnarray*}
&&\lim_{N \rightarrow \infty}\tilde{r}_{10} \\
&=&  \tilde{p}_{11} \log\left(\frac{\tilde{p}_{11}}{\tilde{p}_{1.} \tilde{p}_{.1}}\right) + \tilde{p}_{10} \log\left(\frac{\tilde{p}_{10}}{  \tilde{p}_{1.} \tilde{p}_{.0} }\right) + \tilde{p}_{01} \log\left(\frac{ \tilde{p}_{01}}{ \tilde{p}_{0.} \tilde{p}_{.1}}\right) + \tilde{p}_{00} \log\left(\frac{ \tilde{p}_{00}}{\tilde{p}_{0.} \tilde{p}_{.0}}\right)\\
&=& \log(4) + \tilde{p}_{11} \log(\tilde{p}_{11}) + \tilde{p}_{10} \log(\tilde{p}_{10}) + \tilde{p}_{01} \log(\tilde{p}_{01}) +\tilde{p}_{00} \log(\tilde{p}_{00}) \\
&=& \log(4) + 2 \tilde{p}_{11} \log(\tilde{p}_{11}) + 2\tilde{p}_{10} \log(\tilde{p}_{10}) \\
&=& \log(4) + 2 \tilde{p}_{11} \log(\tilde{p}_{11}) + (1-2\tilde{p}_{11}) \log\left(\frac{1}{2}-\tilde{p}_{11}\right).
\end{eqnarray*}
The proof completes.
\end{proof}

\section{Proof of Theorem~\ref{thm:ratio.scd}}
\begin{proof} \label{proof:ratio.scd}
Since the data $\bX$ is generated from the scenario $\mathcal{S}_{cd}$, the first moment of $X_2$ can be computed as,
\begin{eqnarray*}
\bbE(X_2) &=& \bbE(\bbE(X_2\mid X_1)) = \int_{-\infty}^{+\infty} \bbE(X_2\mid X_1) P(X_1 =x_1) dx_1 \\
&=& \int_{-\infty}^{+\infty} \frac{\exp(\beta(x_1 - \mu_1))}{1+ \exp(\beta(x_1 - \mu_1))} \frac{1}{\sqrt{2\pi} \sigma_1} \exp\left\{ -\frac{(x_1 - \mu_1)^2}{2\sigma^2_1}\right\}dx_1\\
&=& \int_{-\infty}^{+\infty} \frac{1}{\sqrt{2\pi}} \frac{\exp(\beta \sigma_1 t)}{1+ \exp(\beta \sigma_1 t)}  \exp\left\{ -\frac{t^2}{2}\right\}dt.
\end{eqnarray*}
Therefore, the elements of $\Sigma$ are given as
\begin{eqnarray*}
\Sigma_{11} &=& \sigma_1^2\\
\Sigma_{12} &=& \int_{-\infty}^{+\infty} \frac{t \sigma_1}{\sqrt{2\pi}} \frac{\exp(\beta \sigma_1 t)}{1+ \exp(\beta \sigma_1 t)}  \exp\left\{ -\frac{t^2}{2}\right\}dt\\
\Sigma_{22} &=& \bbE(X_2) - (\bbE(X_2))^2,
\end{eqnarray*}
where $\bbE(X_2) = \int_{-\infty}^{+\infty} \frac{\exp(\beta \sigma_1 t)}{1+ \exp(\beta \sigma_1 t)} \frac{1}{\sqrt{2\pi}} \exp\left\{ -\frac{t^2}{2}\right\}dt$. According to the Lemma \ref{lemma:BGe.posterior.ratio}, it is straightforward to derive 
\begin{eqnarray*}
 \lim_{N \rightarrow \infty} r_{10}&=& \frac{1}{2} \log \left( \frac{\Sigma_{11} \Sigma_{22}}{ \Sigma_{11} \Sigma_{22} - \Sigma_{12}^2} \right).
\end{eqnarray*}
For the discretized data $\bX'$, we need to define $\tilde{p}_{ij} \triangleq P(X^{'}_1 = i, X^{'}_2 =j), i \in \{0,1\},j \in \{0,1\}$.
Under the case $\mathcal{S}_{cd}$, we have
\begin{eqnarray}
\tilde{p}_{11} &=& P(X^{'}_1 =1,X^{'}_2=1)  =P(X_1 \geq \mu_1, X_2 =1 ) \nonumber\\
&=&\int_{\mu_1}^\infty \frac{1}{\sqrt{2\pi }\sigma_1  } \exp\left\{- \frac{(x_1-\mu_1)^2}{2 \sigma_1^2} \right\} \frac{\exp\{\beta\times(x_1 - \mu_1)\}}{1 + \exp\{\beta \times(x_1 - \mu_1)\}} dx_1 \nonumber\\
&=& \int_{0}^{\infty} \frac{1}{\sqrt{2\pi } } \exp\left\{- \frac{x^2}{2 } \right\} \frac{\exp\{\beta x \sigma_1 \}}{1+\exp\{\beta x \sigma_1 \}}dx. \label{eqn:p11.scd}
\end{eqnarray}
Similarly, we have 
 \begin{eqnarray} \label{eqn:p00.scd} 
 \tilde{p}_{00} = \tilde{p}_{11}. 
 \end{eqnarray}
Then $\tilde{p}_{10}$ and $\tilde{p}_{01}$ can be derived from
\begin{eqnarray}
&&\tilde{p}_{10} = P( X^{'}_1=1,X^{'}_2=0) =P(X^{'}_1=1)-P(X^{'}_1=1, X^{'}_2=1) =\frac{1}{2}- \tilde{p}_{11}. \label{eqn:p10.scd}\\ 
&& \tilde{p}_{01}=P(X^{'}_1=0,X^{'}_2=1) = P(X^{'}_1=0) - P(X^{'}_1=0, X^{'}_2=0)=\frac{1}{2}-\tilde{p}_{11}. \label{eqn:p01.scd} 
\end{eqnarray}
Given the new data $\bX^{'}$, the posterior ratio between $P(\cG_1 \mid \bX^{'})$ and $P(\cG_0 \mid \bX^{'})$ can be derived by combining Lemma~\ref{lemma:BDe.posterior.ratio} with Equations~\eqref{eqn:p11.scd}, ~\eqref{eqn:p00.scd}, ~\eqref{eqn:p10.scd} and~\eqref{eqn:p01.scd}.
\begin{eqnarray*}
&&\lim_{N \rightarrow \infty} \tilde{r}_{10} \\
&=&  \tilde{p}_{11} \log\left(\frac{\tilde{p}_{11}}{\tilde{p}_{1.} \tilde{p}_{.1}}\right) + \tilde{p}_{10} \log\left(\frac{\tilde{p}_{10}}{  \tilde{p}_{1.} \tilde{p}_{.0} }\right) + \tilde{p}_{01} \log\left(\frac{ \tilde{p}_{01}}{ \tilde{p}_{0.} \tilde{p}_{.1}}\right) + \tilde{p}_{00} \log\left(\frac{ \tilde{p}_{00}}{\tilde{p}_{0.} \tilde{p}_{.0}}\right)\\
&=& \log(4) + \tilde{p}_{11} \log(\tilde{p}_{11}) + \tilde{p}_{10} \log(\tilde{p}_{10}) + \tilde{p}_{01} \log(\tilde{p}_{01}) +\tilde{p}_{00} \log(\tilde{p}_{00}) \\
&=& \log(4) + 2 \tilde{p}_{11} \log(\tilde{p}_{11}) + (1-2\tilde{p}_{11}) \log\left(\frac{1}{2}-\tilde{p}_{11}\right).
\end{eqnarray*}
where $p_{11}$ can be computed as in Equation \eqref{eqn:p11.scd}.
The proof completes.
\end{proof}

\section{Proof of Theorem~\ref{thm:ratio.sdc}}
\begin{proof} \label{proof:ratio.sdc}
Since the data $\bX$ is generated from the case $\mathcal{S}_{dc}$ with $p=\frac{1}{2}$, then $\bbE(X_1) = p=\frac{1}{2}$ and $\Sigma_{11} = p(1-p)=\frac{1}{4}$. To compute $\bbE(X_2)$, we have $\bbE(X_2)= \bbE(\bbE(X_2\mid X_1)) = \bbE(\mu_2 +\beta X_1) = \mu_2+\frac{\beta}{2}$ and  $\bbE(X_2^2)= \bbE(\bbE(X_2^2\mid X_1)) = \sigma_2^2 + \mu_2^2 + \beta \mu_2+\frac{\beta}{2}$. To compute the covariance matrix $\Sigma$, we have
\begin{eqnarray*}
\Sigma_{12}&=& \bbE(X_1 X_2)- \bbE(X_1)\bbE(X_2) =\frac{\beta}{4},\\ 
\Sigma_{22}&=& \beta^2 \sigma_1^2 +\sigma^2_2 = \frac{1}{4} \beta^2 +\sigma_2^2,
\end{eqnarray*}
where $\bbE(X_1 X_2) = \sum_{a\in \{0,1\}} P(X_1=a)\bbE(a X_2 \mid X_1=a) = \frac{\mu_2}{2} + \frac{\beta}{2}$. Therefore, $\Sigma = \left[\begin{array}{cc}
    1/4 & \beta/4  \\
    \beta/4   & \beta^2/4 +\sigma_2^2
\end{array} \right]$. According to the Lemma \ref{lemma:BGe.posterior.ratio}, it is straightforward to derive 
\begin{eqnarray}
\lim_{N\rightarrow \infty}r_{10} = \frac{1}{2} \log \left( \frac{\Sigma_{11} \Sigma_{22}}{\Sigma_{11} \Sigma_{22} - \Sigma_{12}^2} \right) = \frac{1}{2} \log \left( 1+ \frac{\beta^2}{4 \sigma_2^2}\right). \notag 
\end{eqnarray}
For the new data $\bX'$, define $\tilde{p}_{ij} \triangleq P(X^{'}_1 = i, X^{'}_2 =j), i \in \{0,1\},j \in \{0,1\}$.
Under the case $\mathcal{S}_{cd}$, we have
\begin{eqnarray}
&& \tilde{p}_{11} =P(X^{'}_1=1,X^{'}_2=1)  \nonumber\\
&&=P(X_1=1) P( X^{'}_2 =1 \mid X_1 =1) \nonumber\\
&&=\frac{1}{2} \int_{\mu_2 + \frac{\beta}{2}}^\infty \frac{1}{\sqrt{2\pi }\sigma_2  } \exp\left\{- \frac{(x_2-\mu_2-\beta)^2}{2 \sigma_2^2} \right\} dx_2 \nonumber\\
&&=\frac{1}{2}  \int_{-\frac{\beta}{2 \sigma_2}}^{\infty} \frac{1}{\sqrt{2\pi } } \exp\left\{- \frac{x^2}{2 } \right\}dx. \label{eqn:p11.sdc}
\end{eqnarray}
Similarly, we have 
\begin{eqnarray}
 \tilde{p}_{00} = \frac{1}{2} \int_{-\infty}^{\frac{\beta}{2 \sigma_2}} \frac{1}{\sqrt{2\pi } } \exp\left\{- \frac{x^2}{2 } \right\}dx = \tilde{p}_{11}. \label{eqn:p00.sdc}
\end{eqnarray}
Then $\tilde{p}_{10}$ and $\tilde{p}_{01}$ can be derived from
\begin{eqnarray}
&&\tilde{p}_{10} = P( X^{'}_1 =1,X^{'}_2 =0) =P(X^{'}_1 =1)-P(X^{'}_1 =1, X^{'}_2=1) =\frac{1}{2}- \tilde{p}_{11}. \label{eqn:p10.sdc}\\ 
&& \tilde{p}_{01}=P(X^{'}_1 =0,X^{'}_2 =1) = P(X^{'}_1 =0) - P(X^{'}_1 =0, X^{'}_2=0)=\frac{1}{2}-\tilde{p}_{11}. \label{eqn:p01.sdc} 
\end{eqnarray}

Given the new data $\bX^{'}$, the posterior ratio between $P(\cG_1 \mid \bX^{'})$ and $P(\cG_0 \mid \bX^{'})$ can be derived by combining Lemma~\ref{lemma:BDe.posterior.ratio} with Equations \eqref{eqn:p11.sdc}, \eqref{eqn:p00.sdc}, \eqref{eqn:p10.sdc} and \eqref{eqn:p01.sdc}. 
\begin{eqnarray*}
&&\lim_{N \rightarrow \infty} \tilde{r}_{10} \\
&=&  \tilde{p}_{11} \log\left(\frac{\tilde{p}_{11}}{\tilde{p}_{1.} \tilde{p}_{.1}}\right) + \tilde{p}_{10} \log\left(\frac{\tilde{p}_{10}}{  \tilde{p}_{1.} \tilde{p}_{.0} }\right) + \tilde{p}_{01} \log\left(\frac{ \tilde{p}_{01}}{ \tilde{p}_{0.} \tilde{p}_{.1}}\right) + \tilde{p}_{00} \log\left(\frac{ \tilde{p}_{00}}{\tilde{p}_{0.} \tilde{p}_{.0}}\right)\\
&=& \log(4) + \tilde{p}_{11} \log(\tilde{p}_{11}) + \tilde{p}_{10} \log(\tilde{p}_{10}) + \tilde{p}_{01} \log(\tilde{p}_{01}) +\tilde{p}_{00} \log(\tilde{p}_{00}) \\
&=& \log(4) + 2 \tilde{p}_{11} \log(\tilde{p}_{11}) + (1-2\tilde{p}_{11}) \log\left(\frac{1}{2}-\tilde{p}_{11}\right),
\end{eqnarray*}
where $\tilde{p}_{11} = \frac{1}{2}  \int_{-\frac{\beta}{2 \sigma_2}}^{\infty} \frac{1}{\sqrt{2\pi } } \exp\left\{- \frac{x^2}{2 } \right\}dx$.
The proof completes.
\end{proof}

\section{Proof of Theorem~\ref{thm:ratio.sdd}}

\begin{proof} \label{proof:ratio.sdd}

According to the data generation process of $\cS_{dd}$, we have 
\begin{eqnarray*}
\bbE(X_2) &=& p \beta + \frac{1-\beta}{2}, \\
\bbE(X_2^2) &= & \bbE(X_2), \\
\bbV(X_2) &=& \frac{1}{4} - (p - 0.5)^2 \beta^2,\\
\Sigma_{12} &=& \bbE(X_1 X_2)- \bbE(X_1) \bbE(X_2) = \beta p(1-p).
\end{eqnarray*}

Thus, the covariance matrix is given as $$\Sigma = \left[\begin{array}{cc}
    p (1-p) & \beta p (1-p)  \\
    \beta p (1-p)  & \frac{1}{4} - (p - 0.5)^2 \beta^2 
\end{array} \right].$$

Then $$\log\left( \frac{\Sigma_{11} \Sigma_{22}}{\Sigma_{11} \Sigma_{22} - \Sigma_{12}^2}\right) = \log\left( \frac{1- (2p-1)^2 \beta^2}{1- \beta^2} \right).$$


Under the $S_{dd}$ scenario, we have
\begin{eqnarray} \label{eqn:sdd.probability}
p_{11} = p (0.5+\frac{\beta}{2}); p_{10} = p (0.5-\frac{\beta}{2}); p_{01} = (1-p) (0.5-\frac{\beta}{2}); p_{00} = (1-p) (0.5+\frac{\beta}{2}).
\end{eqnarray}
Combining Equation \eqref{eqn:sdd.probability} with the conclusion in Lemma \ref{lemma:BDe.posterior.ratio}, we have
\begin{eqnarray*}
\lim_{N \rightarrow \infty}\tilde{r}_{10} & = & p_{11} \log\left(\frac{p_{11}}{p_{1.} p_{.1}}\right) + p_{10} \log\left(\frac{p_{10}}{  p_{1.} p_{.0} }\right) + p_{01} \log\left(\frac{ p_{01}}{ p_{0.} p_{.1}}\right) + p_{00} \log\left(\frac{ p_{00}}{p_{0.} p_{.0}}\right).
\end{eqnarray*}


The proof completes.
\end{proof}



%% file: main.bbl
\begin{thebibliography}{38}
\newcommand{\enquote}[1]{``#1''}
\expandafter\ifx\csname natexlab\endcsname\relax\def\natexlab#1{#1}\fi
\expandafter\ifx\csname url\endcsname\relax
  \def\url#1{{\tt #1}}\fi
\expandafter\ifx\csname urlprefix\endcsname\relax\def\urlprefix{URL }\fi
\ifx\endbibitem\undefined \let\endbibitem\relax\fi

\bibitem[{Bach and Jordan(2002)}]{bach2002learning}
Bach, F. and Jordan, M. (2002).
\newblock \enquote{Learning graphical models with Mercer kernels.}
\newblock {\em Advances in Neural Information Processing Systems\/}, 15.
\endbibitem

\bibitem[{Boull{\'e}(2006)}]{boulle2006modl}
Boull{\'e}, M. (2006).
\newblock \enquote{{MODL}: {A} {B}ayes optimal discretization method for
  continuous attributes.}
\newblock {\em Machine learning\/}, 65(1): 131--165.
\endbibitem

\bibitem[{Chen et~al.(2017)Chen, Wheeler, and Kochenderfer}]{chen2017learning}
Chen, Y.-C., Wheeler, T.~A., and Kochenderfer, M.~J. (2017).
\newblock \enquote{Learning discrete {B}ayesian networks from continuous data.}
\newblock {\em Journal of Artificial Intelligence Research\/}, 59: 103--132.
\endbibitem

\bibitem[{Chooi et~al.(2019)Chooi, Ding, and Magkos}]{chooi2019epidemiology}
Chooi, Y.~C., Ding, C., and Magkos, F. (2019).
\newblock \enquote{The epidemiology of obesity.}
\newblock {\em Metabolism\/}, 92: 6--10.
\endbibitem

\bibitem[{Friedman and Goldszmidt(1996)}]{friedman1996discretizing}
Friedman, N. and Goldszmidt, M. (1996).
\newblock \enquote{Discretizing continuous attributes while learning {B}ayesian
  {N}etworks.}
\newblock In {\em ICML\/}, 157--165.
\endbibitem

\bibitem[{Friedman and Koller(2003)}]{friedman2003structuremcmc}
Friedman, N. and Koller, D. (2003).
\newblock \enquote{Being {B}ayesian {A}bout {N}etwork {S}tructure. {A B}ayesian
  {A}pproach to {S}tructure {D}iscovery in {B}ayesian {N}etworks.}
\newblock {\em Machine Learning\/}, 50: 95--125.
\endbibitem

\bibitem[{Geiger and Heckerman(1994)}]{geiger1994learning}
Geiger, D. and Heckerman, D. (1994).
\newblock \enquote{{Learning Gaussian Networks}.}
\newblock In {\em Uncertainty Proceedings 1994\/}, 235--243. Elsevier.
\endbibitem

\bibitem[{Geiger and Heckerman(2002)}]{geiger2002parameter}
--- (2002).
\newblock \enquote{Parameter priors for directed acyclic graphical models and
  the characterization of several probability distributions.}
\newblock {\em The Annals of Statistics\/}, 30(5): 1412--1440.
\endbibitem

\bibitem[{Glymour et~al.(2019)Glymour, Zhang, and Spirtes}]{glymour2019review}
Glymour, C., Zhang, K., and Spirtes, P. (2019).
\newblock \enquote{Review of causal discovery methods based on graphical
  models.}
\newblock {\em Frontiers in genetics\/}, 10: 524.
\endbibitem

\bibitem[{Grzegorczyk and Husmeier(2008)}]{grzegorczyk2008improving}
Grzegorczyk, M. and Husmeier, D. (2008).
\newblock \enquote{Improving the structure {MCMC} sampler for {B}ayesian
  networks by introducing a new edge reversal move.}
\newblock {\em Machine Learning\/}, 71(2-3): 265.
\endbibitem

\bibitem[{Heckerman and Geiger(1995)}]{heckerman1995}
Heckerman, D. and Geiger, D. (1995).
\newblock \enquote{Learning {Bayesian Networks: A Unification for Discrete and
  Gaussian Domains}.}
\newblock In {\em Proceedings of the Eleventh Conference on Uncertainty in
  Artificial Intelligence\/}, UAI'95, 274--284. San Francisco, USA.
\endbibitem

\bibitem[{Heckerman et~al.(1995)Heckerman, Geiger, and
  Chickering}]{heckerman1995learning}
Heckerman, D., Geiger, D., and Chickering, D.~M. (1995).
\newblock \enquote{Learning {B}ayesian networks: The combination of knowledge
  and statistical data.}
\newblock {\em Machine learning\/}, 20(3): 197--243.
\endbibitem

\bibitem[{Heinze-Deml et~al.(2018)Heinze-Deml, Maathuis, and
  Meinshausen}]{heinze2018causal}
Heinze-Deml, C., Maathuis, M.~H., and Meinshausen, N. (2018).
\newblock \enquote{Causal {s}tructure {l}earning.}
\newblock {\em Annual Review of Statistics and Its Application\/}, 5: 371--391.
\endbibitem

\bibitem[{Koller and Friedman(2009)}]{koller2009probabilistic}
Koller, D. and Friedman, N. (2009).
\newblock {\em Probabilistic graphical models: principles and techniques\/}.
\newblock MIT press.
\endbibitem

\bibitem[{Kuipers and Moffa(2017)}]{kuipers2017partition}
Kuipers, J. and Moffa, G. (2017).
\newblock \enquote{Partition {MCMC} for inference on acyclic digraphs.}
\newblock {\em Journal of the American Statistical Association\/}, 112(517):
  282--299.
\endbibitem

\bibitem[{Kuipers et~al.(2014)Kuipers, Moffa, and
  Heckerman}]{kuipers2014addendum}
Kuipers, J., Moffa, G., and Heckerman, D. (2014).
\newblock \enquote{Addendum on the scoring of {G}aussian directed acyclic
  graphical models.}
\newblock {\em The Annals of Statistics\/}, 42(4): 1689--1691.
\endbibitem

\bibitem[{Lauritzen(1992)}]{lauritzen1992propagation}
Lauritzen, S.~L. (1992).
\newblock \enquote{Propagation of probabilities, means, and variances in mixed
  graphical association models.}
\newblock {\em Journal of the American Statistical Association\/}, 87(420):
  1098--1108.
\endbibitem

\bibitem[{Lauritzen and Wermuth(1989)}]{lauritzen1989graphical}
Lauritzen, S.~L. and Wermuth, N. (1989).
\newblock \enquote{Graphical models for associations between variables, some of
  which are qualitative and some quantitative.}
\newblock {\em The Annals of Statistics\/}, 17(1): 31--57.
\endbibitem

\bibitem[{Lorch et~al.(2021)Lorch, Rothfuss, Sch{\"o}lkopf, and
  Krause}]{lorch2021dibs}
Lorch, L., Rothfuss, J., Sch{\"o}lkopf, B., and Krause, A. (2021).
\newblock \enquote{Dibs: Differentiable bayesian structure learning.}
\newblock {\em Advances in Neural Information Processing Systems\/}, 34:
  24111--24123.
\endbibitem

\bibitem[{Lustgarten et~al.(2011)Lustgarten, Visweswaran, Gopalakrishnan, and
  Cooper}]{lustgarten2011application}
Lustgarten, J.~L., Visweswaran, S., Gopalakrishnan, V., and Cooper, G.~F.
  (2011).
\newblock \enquote{Application of an efficient {B}ayesian discretization method
  to biomedical data.}
\newblock {\em BMC bioinformatics\/}, 12(1): 1--15.
\endbibitem

\bibitem[{Madigan et~al.(1995)Madigan, York, and Allard}]{madigan1995bayesian}
Madigan, D., York, J., and Allard, D. (1995).
\newblock \enquote{{B}ayesian graphical models for discrete data.}
\newblock {\em International Statistical Review/Revue Internationale de
  Statistique\/}, 63(2): 215--232.
\endbibitem

\bibitem[{Mohal et~al.(2021)Mohal, Lansangan, Gasser, Taylor, Renda, Jessup,
  and Daraganova}]{mohal2021}
Mohal, J., Lansangan, C., Gasser, C., Taylor, T., Renda, J., Jessup, K., and
  Daraganova, G. (2021).
\newblock \enquote{Growing Up in Australia: The Longitudinal Study of
  Australian Children – Data User Guide, Release 9C1.}
\newblock Technical report, Melbourne: Australian Institute of Family Studies.
\endbibitem

\bibitem[{Monti and Cooper(2013)}]{monti2013multivariate}
Monti, S. and Cooper, G.~F. (2013).
\newblock \enquote{A multivariate discretization method for learning {B}ayesian
  networks from mixed data.}
\newblock {\em arXiv preprint arXiv:1301.7403\/}.
\endbibitem

\bibitem[{Moral et~al.(2001)Moral, Rum{\'\i}, and
  Salmer{\'o}n}]{moral2001mixtures}
Moral, S., Rum{\'\i}, R., and Salmer{\'o}n, A. (2001).
\newblock \enquote{Mixtures of truncated exponentials in hybrid {B}ayesian
  networks.}
\newblock In {\em European Conference on Symbolic and Quantitative Approaches
  to Reasoning and Uncertainty\/}, 156--167. Springer.
\endbibitem

\bibitem[{Nemes(2010)}]{Stirling2010}
Nemes, G. (2010).
\newblock \enquote{On the Coefficients of the Asymptotic Expansion of $n!$.}
\newblock {\em Journal of Integer Sequences\/}, 13(2): 3.
\endbibitem

\bibitem[{Nojavan et~al.(2017)Nojavan, Qian, and Stow}]{nojavan2017comparative}
Nojavan, F., Qian, S.~S., and Stow, C.~A. (2017).
\newblock \enquote{Comparative analysis of discretization methods in {B}ayesian
  networks.}
\newblock {\em Environmental Modelling \& Software\/}, 87: 64--71.
\endbibitem

\bibitem[{Rissanen(1978)}]{rissanen1978modeling}
Rissanen, J. (1978).
\newblock \enquote{Modeling by shortest data description.}
\newblock {\em Automatica\/}, 14(5): 465--471.
\endbibitem

\bibitem[{Robinson(1977)}]{1978DAGn}
Robinson, R.~W. (1977).
\newblock \enquote{Counting unlabeled acyclic digraphs.}
\newblock In Little, C. H.~C. (ed.), {\em Combinatorial Mathematics V\/},
  28--43. Berlin, Heidelberg: Springer Berlin Heidelberg.
\endbibitem

\bibitem[{Sokolova et~al.(2014)Sokolova, Groot, Claassen, and
  Heskes}]{sokolova2014causal}
Sokolova, E., Groot, P., Claassen, T., and Heskes, T. (2014).
\newblock \enquote{Causal discovery from databases with discrete and continuous
  variables.}
\newblock In {\em European Workshop on Probabilistic Graphical Models\/},
  442--457. Springer.
\endbibitem

\bibitem[{Sun and Erath(2015)}]{sun2015bayesian}
Sun, L. and Erath, A. (2015).
\newblock \enquote{A {B}ayesian network approach for population synthesis.}
\newblock {\em Transportation Research Part C: Emerging Technologies\/}, 61:
  49--62.
\endbibitem

\bibitem[{Tavana et~al.(2018)Tavana, Abtahi, Di~Caprio, and
  Poortarigh}]{tavana2018artificial}
Tavana, M., Abtahi, A.-R., Di~Caprio, D., and Poortarigh, M. (2018).
\newblock \enquote{{An Artificial Neural Network and Bayesian Network model for
  liquidity risk assessment in banking}.}
\newblock {\em Neurocomputing\/}, 275: 2525--2554.
\endbibitem

\bibitem[{Van De~Schoot et~al.(2017)Van De~Schoot, Winter, Ryan,
  Zondervan-Zwijnenburg, and Depaoli}]{van2017systematic}
Van De~Schoot, R., Winter, S.~D., Ryan, O., Zondervan-Zwijnenburg, M., and
  Depaoli, S. (2017).
\newblock \enquote{A systematic review of {B}ayesian articles in psychology:
  The last 25 years.}
\newblock {\em Psychological Methods\/}, 22(2): 217.
\endbibitem

\bibitem[{Vowels et~al.(2021)Vowels, Camgoz, and Bowden}]{vowels2021d}
Vowels, M.~J., Camgoz, N.~C., and Bowden, R. (2021).
\newblock \enquote{D’ya like {DAG}s? {A} {s}urvey on {s}tructure {l}earning
  and {c}ausal {d}iscovery.}
\newblock {\em ACM Computing Surveys (CSUR)\/}.
\endbibitem

\bibitem[{Xing et~al.(2017)Xing, Guo, Liu, Wang, Wang, and Zhang}]{xing2017GRN}
Xing, L., Guo, M., Liu, X., Wang, C., Wang, L., and Zhang, Y. (2017).
\newblock \enquote{{An improved Bayesian network method for reconstructing gene
  regulatory network based on candidate auto selection}.}
\newblock {\em BMC Genomics\/}, 18(9): 17--30.
\endbibitem

\bibitem[{Yu et~al.(2019)Yu, Chen, Gao, and Yu}]{yu2019dag}
Yu, Y., Chen, J., Gao, T., and Yu, M. (2019).
\newblock \enquote{DAG-GNN: DAG structure learning with graph neural networks.}
\newblock In {\em International Conference on Machine Learning\/}, 7154--7163.
  PMLR.
\endbibitem

\bibitem[{Zhang et~al.(2018)Zhang, B{\"u}tepage, Kjellstr{\"o}m, and
  Mandt}]{zhang2018advances}
Zhang, C., B{\"u}tepage, J., Kjellstr{\"o}m, H., and Mandt, S. (2018).
\newblock \enquote{{Advances in Variational Inference}.}
\newblock {\em IEEE transactions on Pattern Analysis and Machine
  Intelligence\/}, 41(8): 2008--2026.
\endbibitem

\bibitem[{Zhang et~al.(2019)Zhang, Jiang, Cui, Garnett, and Chen}]{zhang2019d}
Zhang, M., Jiang, S., Cui, Z., Garnett, R., and Chen, Y. (2019).
\newblock \enquote{D-vae: A variational autoencoder for directed acyclic
  graphs.}
\newblock {\em Advances in Neural Information Processing Systems\/}, 32.
\endbibitem

\bibitem[{Zheng et~al.(2018)Zheng, Aragam, Ravikumar, and Xing}]{zheng2018dags}
Zheng, X., Aragam, B., Ravikumar, P.~K., and Xing, E.~P. (2018).
\newblock \enquote{{Dags with NOTEARS: Continuous optimization for structure
  learning}.}
\newblock {\em Advances in Neural Information Processing Systems\/}, 31.
\endbibitem

\end{thebibliography}
